\documentclass[12pt]{article}
\usepackage{setspace}
\usepackage{amsmath}
\usepackage{amssymb}
\usepackage{amsfonts}
\usepackage{amsthm}
\usepackage{eurosym}
\usepackage{lastpage}
\usepackage{bbm}
\usepackage{stackengine}
\usepackage{subcaption}
\usepackage{fancyhdr}
\usepackage{graphicx,color}
\usepackage{setspace}
\usepackage{multirow}
\usepackage{xcolor}
\usepackage{tabularx}
\usepackage{wrapfig}  
\usepackage{longtable}
\usepackage{threeparttable}
\usepackage{bm}
\usepackage{lipsum}
\usepackage{algorithm}
\usepackage{algorithmic}
\usepackage[euler-digits]{eulervm}
\usepackage[font=small,font=bf]{caption}
\definecolor{brightmaroon}{RGB}{131,3,0}
\usepackage[authoryear]{natbib} %
\usepackage[colorlinks = true,
            backref=page,
            citecolor  = brightmaroon,
            linkcolor  = blue,
            urlcolor   = blue]{hyperref}
\usepackage[a4paper, portrait, margin=1in]{geometry}

\graphicspath{{./Figures/}}

\usepackage{mathtools}
\mathtoolsset{showonlyrefs}

\usepackage{tikz}
\usetikzlibrary{positioning,chains,fit,shapes,calc,arrows,backgrounds,3d,matrix,arrows.meta}

\usepackage{pifont}
\usepackage{diagbox}
\definecolor{Red}{RGB}{255,0,0}
\definecolor{ForestGreen}{RGB}{34, 139, 34} %
\newcommand{\mytick}{{\color{ForestGreen}\ding{52}}}
\newcommand{\mycross}{{\color{Red}\ding{56}}}

\DeclareSymbolFont{Symbols}{OMS}{zplm}{m}{n}%
\DeclareMathSymbol{\Infty}{\mathord}{Symbols}{"31}
\usepackage{mathpazo} %
\usepackage{array}
\newlength{\noteslabelwidth}
\renewenvironment{tablenotes}{%
  \par\vspace{-0.5\baselineskip}%
  \footnotesize
  \justifying
  \singlespacing
  \noindent
 \settowidth{\noteslabelwidth}{\textbf{Notes:}}%
  \begin{tabular}{p{\dimexpr\linewidth\relax}@{}}
  \textbf{Notes:} 
}{%
  \end{tabular}

  \par
}

\renewcommand*{\backrefalt}[4]{%
    \ifcase #1%
     \or (Located on page:~#2)%
     \else (Located on page:~#2)%
    \fi%
    } 

\usepackage{ragged2e}

\renewcommand{\ni}{\noindent}
\renewcommand{\Vec}{\text{vec}}

\newtheorem{lem}{Lemma}
\newtheorem{ass}{Assumption}
\newtheorem{dfn}{Definition}

\makeatletter
\def\@fnsymbol#1{\ensuremath{\ifcase#1\or *\or \mathsection\or \mathparagraph\or \or \dagger \ddagger\or \| \else\@ctrerr\fi}}
\makeatother

\title{Can looser ties sustain marriage? A dynamic matching model of specialisation and divorce\thanks{\hangindent=.66cm For helpful comments and discussions, we thank James Banks, Pauline Corblet, Edoardo Ciscato, Alejandro Robinson-Cort\'es, Jeremy Fox, Alfred Galichon, Lars Nesheim, David Pacini, Pietro Spini, Matt Shum, and all participants of seminars and workshops, including the Exeter Matching Markets and Inequality Workshop, Bristol Econometric Study Group, SEHO, NASM, ESWC, EWMES, IAAE, AGEW, ESPE, and the Bristol and Melbourne research seminar. This research has been supported by a Faculty Research Grant funding provided by the University of Melbourne. Any mistakes are our own.\bigskip} }

\usepackage[blocks]{authblk}
\author{%
\begin{minipage}[t]{0.45\textwidth}
\centering
Stefan Hubner\thanks{Department of Economics, \url{stefan.hubner@bristol.ac.uk}}\\
University of Bristol
\end{minipage}
\hfill
\begin{minipage}[t]{0.45\textwidth}
\centering
Jan Kabatek\thanks{Melbourne Institute of Applied Economic and Social Research, \url{j.kabatek@unimelb.edu.au}}\\
University of Melbourne
\end{minipage}
}

\newif\ifbeamer
\beamerfalse

\date{June 2026} 

\makeatletter%
\begin{document} 

\maketitle

\renewcommand{\thefootnote}{\arabic{footnote}}
\setcounter{footnote}{0}

\vspace{-\baselineskip}
\begin{abstract}
\ni Durable marriages are presumed to foster the household specialisation that marriage enables.
We exploit a recent Dutch reform that temporarily lowered the cost of divorce while leaving consent requirements unchanged.
We embed divorce hazards obtained from population-level administrative data into a dynamic structural matching model in which individuals repeatedly match and choose marital roles.
We identify the structural parameters by fitting the model to the equilibrium matching distribution over time, using a novel computational approach.
Compared to the high-cost counterfactual, we find that more couples choose marriage when divorce costs are lower, as higher rates of marriage entry outweigh the rise in divorce.
Because specialisation is preserved, aggregate welfare rises. 

\end{abstract}

\hangindent=2.75cm Keywords: divorce costs, marriage markets, marriage contracts, collective choice, random utility models, dynamic programming %

\hangindent=2.75cm JEL codes: J12, C54, C63, C78. \bigskip

\thispagestyle{empty}

\newpage

\newpage \setcounter{page}{1}
\section{Introduction} 

Marriage rates in the United States and Western Europe have fallen by more than 30\% over the last thirty years, with couples increasingly opting for cohabitation out of wedlock \citep{oecd2025marriage}.
This is consequential, because many cohabiting couples forgo the protections and benefits associated with marriage.
The literature \citep{Stevenson2007} attributes the falling popularity of marriage to broader structural changes in the marriage market, including changes in fertility control, household
production, and other secular trends.
We show that the overall cost of marriage, which has received less attention, is also an important contributor towards this trend.

Marriages come with substantial financial costs. These include the costs of the wedding, but also the costs of divorce. Indeed, divorces can be very expensive even when both spouses consent, and  
this may deter couples from formalising their relationships. A policymaker might therefore consider lowering the divorce costs in an attempt to foster marriage.\footnote{The administrative cost for the civil marriage ceremony is already low (\euro 100--\euro 200), with discretionary spending on the celebration varying widely with couples' preferences.} 
However, lower costs of divorce may also encourage existing couples to break up, which renders the net effect on marriage incidence theoretically ambiguous.\footnote{The two effects follow the insider-outsider logic of \citet{Lindbeck2001}: lower costs reduce the expected price of entry for outsiders (the unmarried) and the price of exit for insiders (the married).}
On top of this, even if the overall number of marriages remained unaffected, easier divorce may lower aggregate welfare by weakening marital investments (specialisation) and its associated positive externalities \citep{Becker1973}.

In this paper, we investigate what happens to marriage, and to the welfare it generates, when divorce costs are lowered.
Studying this question is difficult for two reasons. First, married couples are rarely subject to exogenous changes in divorce costs. 
Indeed, most divorce studies analyse a shift in consent rules rather than prices (\textit{e.g.}, \citealp*{wolfers2006,fernandez2017,Calvo2025}).
We exploit an accidental change of the Dutch family law that produced a large and long-lasting reduction in the effective costs of divorce. The change in question was the 2001 legalisation of the administrative divorce option. Colloquially named ``flash divorce'', this option enabled Dutch couples to settle their divorces out of court, waiving the requirement of legal representation and reducing the associated costs, often by several thousands of Euros. 
After eight years, in 2009, the Dutch government abolished the administrative option, reverting to the original high-cost regime. This unique policy variation enables us to investigate the marriage market responses to a large and plausibly exogenous change in  divorce costs. 
The second difficulty is that the welfare consequences of a price change cannot be identified solely from the responses of marginal couples. 
Half of the divorcees eventually re-partner, which shifts the outside option of every couple still together, married or cohabiting. This can have substantive consequences for the equilibrium structure of the marriage market, with the full effects of the price change manifesting only after a number of years. It is also worth noting that, for the divorcees, the very roles that might have created positive externalities within their marriage may now work against them.\footnote{Indeed, while the human capital accrued through market work may act akin to college premium (increasing the attractiveness of a person re-entering the marriage market), the human capital accrued through household specialisation likely depreciates to a large extent once the couple divorces.}\textsuperscript,\footnote{For a study on collective bargaining after re-marriage see \citet{Bonke2009a,Bonke2009}.} %
Our setting enables to account for these marriage market features. This is because we study the markets through the lens of repeated matching, where individuals are described by their \emph{marital status} (single, cohabiting, or married) and their \emph{household roles} (market work or home production). 
Both characteristics are endogenous and the result of decisions based on preference primitives and the market environment.

We observe the equilibrium outcomes: who is paired with whom, who is single, who has specialised in which role, and how these change over time.
As a classical inverse problem \citep{Becker1973}, we ask which preferences rationalise this equilibrium path, which we obtain from a rich micro-level dataset constructed from linked population registers.
We build a structural model to study the dynamics of the Dutch marriage market and evaluate its equilibrium adaptations induced by the changes of divorce costs.
Three sources of variation identify counterfactual trajectories of our economy.
First, the cross-sectional composition of couples across types identifies the surplus generated by each match type. 
Second, the 2009 ban identifies the drop in divorce hazards for any of the married couple's types described by the respective partner's roles which we relate to the insider surplus of staying in each type of marriage.
Third, individual-level transitions identify the type-transition kernel non-parametrically.\footnote{Repeated cross section and even panel surveys are generally not sufficient for this. Indeed, even with large samples, transitions cannot be identified non-parametrically at the resolution of the type space because cells become sparse quickly. With sampling uncertainty of the transition kernel, a zero might only be a sampling artifact but would shut down an entire transition path. This can be ruled out with administrative population-level data. In addition, we do not lose either ex-spouse after a divorce as is often the case in surveys.}

The model treats divorce as a matching decision of insiders against the outside option, and marriage as the inflow of outsiders which is governed by the transition kernels. This separation of divorce and marriage allows us to disentangle the insider and outsider channels.
Their separate identification comes from the reform, which shifts the cost of divorce abruptly, while the stock of matches that identifies surplus adjusts only gradually. 
Exploiting this difference in timing requires both the reform and the administrative data, and must be operationalised through a dynamic model without steady-state restrictions.\footnote{For example, without transitions the surplus would have to be recovered jointly from the matching data which only identifies net flows. Disentangling them, as in \citet*{Gousse2017}, requires a steady-state restriction which is ruled out by the nature of our reform. Note that, even if there is no effect of the reform on divorce rates, the marriage market is still in constant transition.} 

We find that lower divorce costs induce more divorces among those who are already married, but they induce even more marriage entries, yielding a positive net effect on marriage incidence.
We find no evidence of a decline in specialisation among those who remain married. A counterfactual simulation shows that the long-term aggregate welfare would have been 17 percent higher, if the administrative divorce option remained in place.
While dual earner couples respond to the reform, male-breadwinner couples remain largely insulated on the divorce margin, which reflects the complexity of their dissolutions (due to alimony and property division).

The literature on the welfare effects of divorce is largely a literature on unilateral-divorce reform. There, the change in family law shifts the right to dissolve from one spouse to the other and operates through bargaining inside the household. \citet*{Chiappori2002}, \citet{mazzocco2007}, \citet{Voena2015}, \citet{fernandez2017}, and \citet*{LowMeghirPistaferriVoena2018} trace the consequences through intra-household allocation, savings, and labour supply. In these models divorce is exit into a singles value function. 
Building on \citet*{ChiapporiCostaDiasMeghir2018}, \citet{reynoso2024} embed the collective household inside a marriage market with endogenous sorting, and show that family law reshapes who marries whom. Their matching is one-shot, on premarket education types, and divorcees do not return to it.

Using observational data, \citet*{Bruze2015} estimate a finite-horizon cohort-life-cycle model of marriage choices using Danish administrative data. They use flows of marriage and divorce to identify marital surplus.
\citet{Flinn2026} consider a dynamic continuous-time model of marriage and labour supply decisions.

We consider a reform that is a change in price, not in rights.
It operates through the marriage market rather than through intra-household bargaining.
Thus, we identify marital surplus through a matching model with transferable utility (\citealp{Choo2006}, \citealp{GalichonSalanie2022}).
We incorporate the divorce-hazards from the regression discontinuity into their two-sided logit framework (Lemma \ref{lem:hazard}). Because  the type-specific distribution of unobserved preference parameters is identified directly from the cross-section of matches we can recover the surplus shift associated with divorce costs for each marriage type.
The repeated nature of our model, which captures the induced long-run market adjustments, puts us close to \citet{ChenChoo2023}, with our types being endogenous and model-determined with dynamics as in \citet*{Corblet2023}, whose planner-equilibrium equivalence result we rely on.

Since both the permanent dynamics of the marriage market and the transient nature of our reform rule out identification based on a steady-state equilibrium we have to solve the full dynamic programme.
We identify the surplus parameters by minimising the distance (Kullback-Leibler-divergence) between the model-implied and empirical conditional distributions of couples.\footnote{Conditioning on the marginals insulates the estimator from the accumulation of small errors in the transitions between types.}
Solving the dynamic programme in each step is computationally difficult due to the continuous state space describing population masses across the type space.
Thus, we approximate the value and policy functions jointly with a neural network \citep*{maliar2021}.
We enforce the market-clearing conditions through a \citet{sinkhorn1964} matrix scaling in the final layer.
We, further, introduce a penalisation of the approximation error of the value function to make the algorithm adhere to the model's dynamic structural restrictions.

The paper is structured around its identification logic. Section \ref{sec:policy} describes the institutional context and the policy that separates the two divorce-cost regimes. Section \ref{Economic Theory} develops the structural model. Section \ref{Data and Descriptive Statistics} presents the data that identify the model. Section \ref{Results} discusses our results and counterfactual simulations. Section \ref{Conclusions} concludes.

\section{Institutional context}\label{sec:policy}

Before we dive into the theory and data sections, we present the key institutional features of the Dutch family law system. Further details can be found in Appendix \ref{sec:institutions}.

\subsection{Marriage and Divorce in the Netherlands}

Marriage in the Netherlands follows patterns broadly comparable to those observed in other developed countries. As shown in Figure \ref{fig_mar_rate}, the Dutch marriage rate is similar to those in other western European countries and the United Kingdom. Its evolution over time is also comparable: over the 20-year period of observation, marriage rates in the plotted countries fell by nearly 30\%, consistent with the broader de-institutionalisation of marriage \citep{Cherlin2004, Stevenson2007}.

\begin{figure}[h!]
\caption{Marriage rates in the Netherlands, UK, and Western Europe: 1995--2016}
\centering
\includegraphics[trim={0 0 0 0},clip,width=0.62\textwidth]{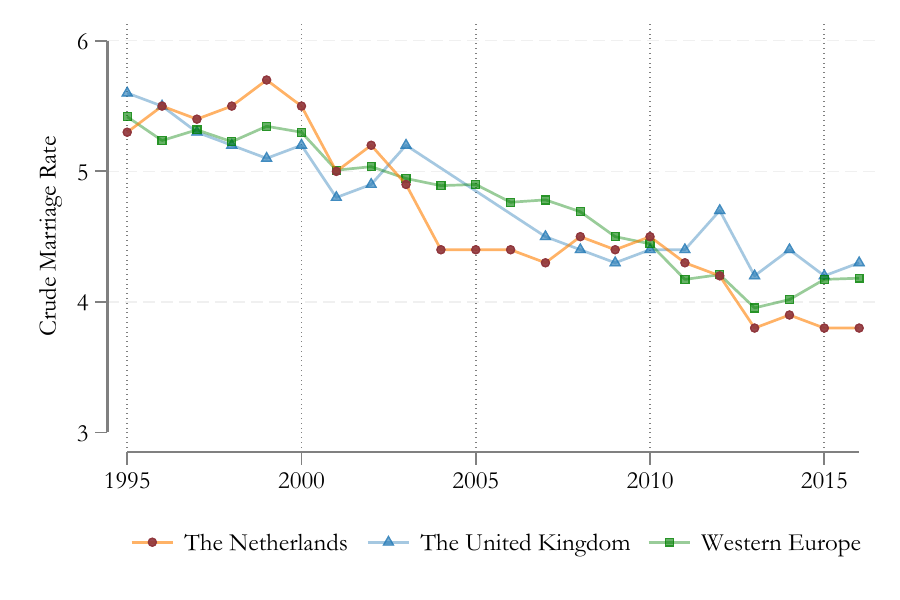}
\label{fig_mar_rate}
\vspace{-0.2cm}
\begin{tablenotes}
Data from OECD. Rates represent crude marriage rate: the number of marriages commenced in the given year per 1000 inhabitants.  
\end{tablenotes}
\end{figure} 

From the procedural perspective, Dutch couples enter marriage by registering their intent with the municipality and participating in a civil ceremony. Divorce, by contrast, requires a petition to the family court, and it is only granted once the court determines that the marriage is irretrievably broken. Due to the involvement of the judicial system, divorce is associated with substantial costs. Couples are required to retain legal representation and pay fees for court filings and notary services. This means that, even for an uncontested divorce, the costs can amount to \EUR{3,000}-\EUR{5,000}.\footnote{Low-income couples may qualify for subsidised legal aid and discounted court fees.}

These costs are comparable to the costs observed in countries where court proceedings are also mandatory (\textit{e.g.,} most European countries, Canada, and Australia).\footnote{Reliable statistics on divorce costs are not systematically collected, which makes rigorous cross-country comparisons difficult. The discussion in this section is based on information retrieved from family law discussion boards, practitioner guides, and indicative fee schedules published by law firms and legal service providers in the respective countries.} In countries where court proceedings are not mandatory (e.g., Scandinavia, the Baltics, Portugal, Spain, and some U.S. states), divorce costs can be substantially lower, in some cases as low as \EUR{80}. The markedly lower costs in the latter group stem from the availability of the administrative divorce option. This procedure allows couples to dissolve their marriages without the need to engage lawyers or appear before a judge. Aimed at couples who agree on all aspects of their union dissolution, it reduces the divorce process to filling out a simple standardised form to update the marriage registry at the local municipality.

\subsection{Accidental emergence of the administrative divorce option}

Central to our study is a unique policy experiment: the accidental introduction and eventual repeal of the administrative divorce option in the Netherlands. This experiment was an unintended consequence of the 2001 Marriage Equality Act, a legislation that enabled the Netherlands to become the first country to legalise same-sex marriage.  

Besides its intended purpose, the Marriage Equality Act created a legal loophole that allowed married couples to exercise the administrative divorce option.
The option quickly gained popularity and became known as the \textit{flitsscheiding} (``flash divorce''), a nickname that reflected its simplicity compared to the conventional divorce procedure. The conventional procedure remained unaffected by the Act, which meant that, from April 2001 onward, married couples in the Netherlands could choose between the judicial and administrative divorce pathways.
Though it emerged by accident, the dual regime was maintained for eight years, coming to an end only in April 2009, when the administrative option was rescinded through an amendment of the Dutch civil code. This change was not instigated by public demand or a landmark legal case, but rather represented a return to the legislative status quo \citep{Aygun2015}.

More than 30,000 couples had taken advantage of the administrative option over the span of its availability.\footnote{This is according to our own calculations, consistent with the report of \citet{statnl2010}.}
Figure \ref{fig_div_inc_NL} demonstrates that this uptake did not reflect a mere procedural substitution, where couples who would have divorced anyway were simply opting for the cheaper procedure. Rather, the divorce rates during the period of flash divorce availability were, on average, approximately 10 percent higher than the divorce rates in the surrounding years.

\begin{figure}[h!]
\caption{Crude divorce rates in the Netherlands: 1995-2016}
\includegraphics[trim={0 0 0 0},clip,width=\textwidth]{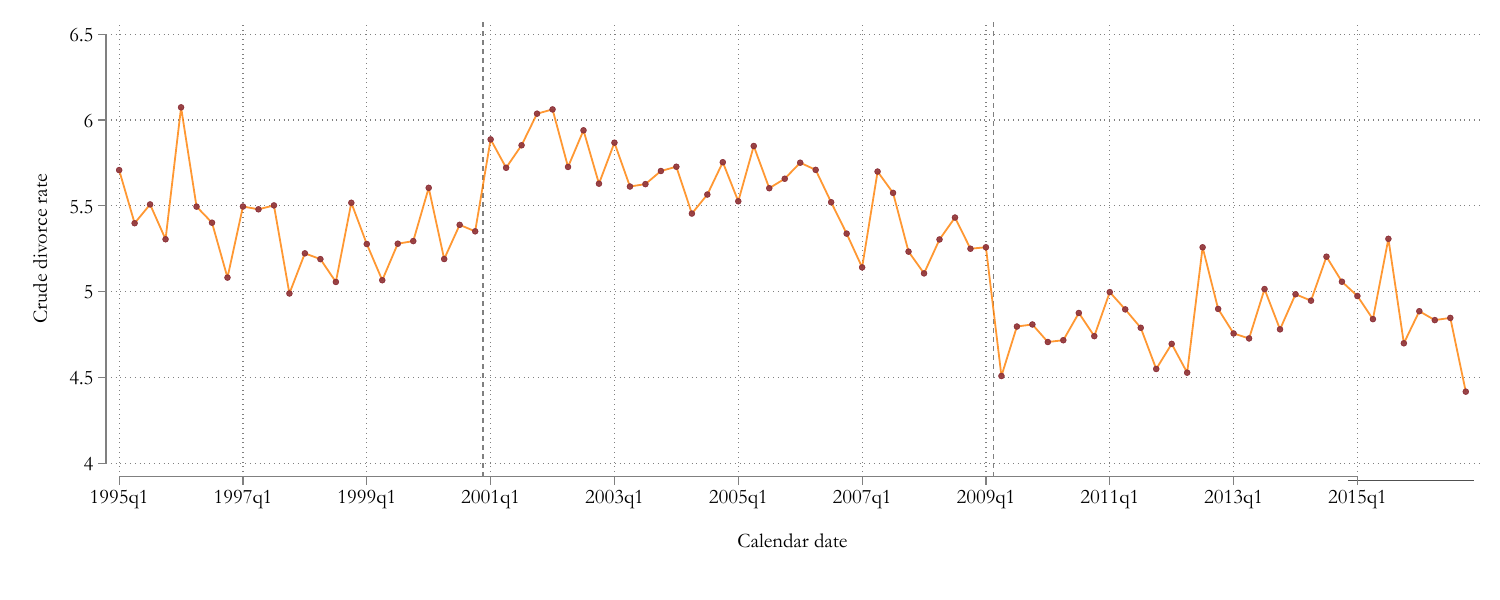}
\label{fig_div_inc_NL}
\vspace{-0.8cm}
\begin{tablenotes}
Data from Statistics Netherlands (CBS). Crude divorce rate represents the number of divorces per 1000 inhabitants. Plotted quarterly rates were annualised to facilitate exposition. Introduction and repeal of the administrative divorce option are highlighted by the dashed vertical lines.
\end{tablenotes}
\end{figure} 

Acknowledging there could be macroeconomic factors confounding this series, we compare the Netherlands to culturally similar controls.
Figure \ref{fig_div_inc_EUR} plots the raw national divorce rates of the selected countries.\footnote{The large increase of divorce rates recorded in Denmark in the years 2013 and 2014 (blue line) is attributed to the introduction of a very similar administrative divorce option \citep{Fallesen2021}. The subsequent downward adjustment recorded in 2015 is considered to be an artefact of procedural changes in the Danish marriage registry.}
Figure \ref{fig_div_inc_EUR_NL} reports the Dutch divorce rates against a synthetic control of those countries with uniform weights.
There is clear excess divorce in the Netherlands within the low-cost period, with no notable discontinuities in the control group around the two respective cut-off dates.

\begin{figure}[h!]
\caption{Crude divorce rates in select European countries: 1995 -- 2016}
\begin{subfigure}{0.48\textwidth}
\vspace{0pt}
\includegraphics[trim={18 0 0 0},clip,width=\textwidth,height=4.75cm]{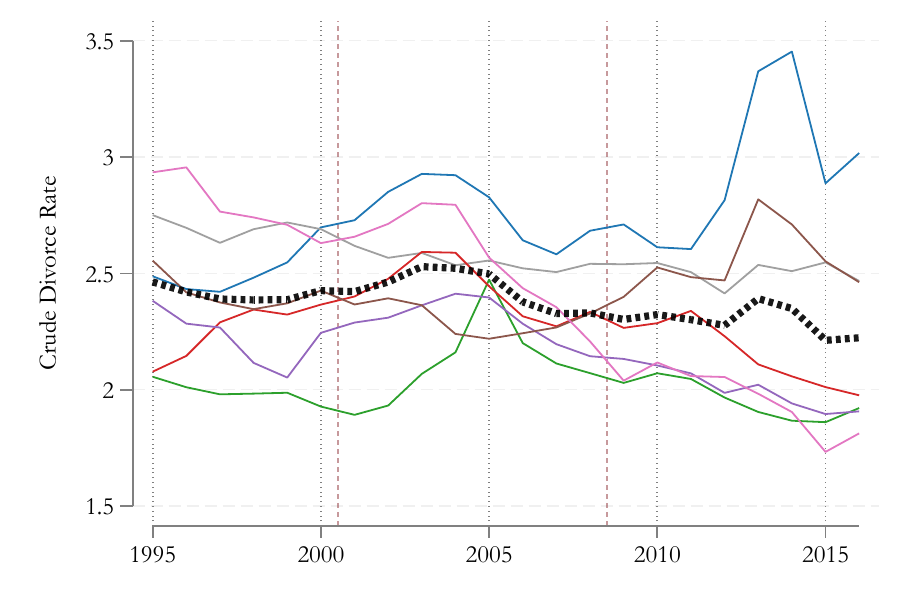}
\caption{Selected European Countries}
\label{fig_div_inc_EUR}
\end{subfigure}
\hfill
\begin{subfigure}{0.48\textwidth}
\vspace{0pt}
\includegraphics[trim={18 0 0 0},clip,width=\textwidth,height=5cm,keepaspectratio]{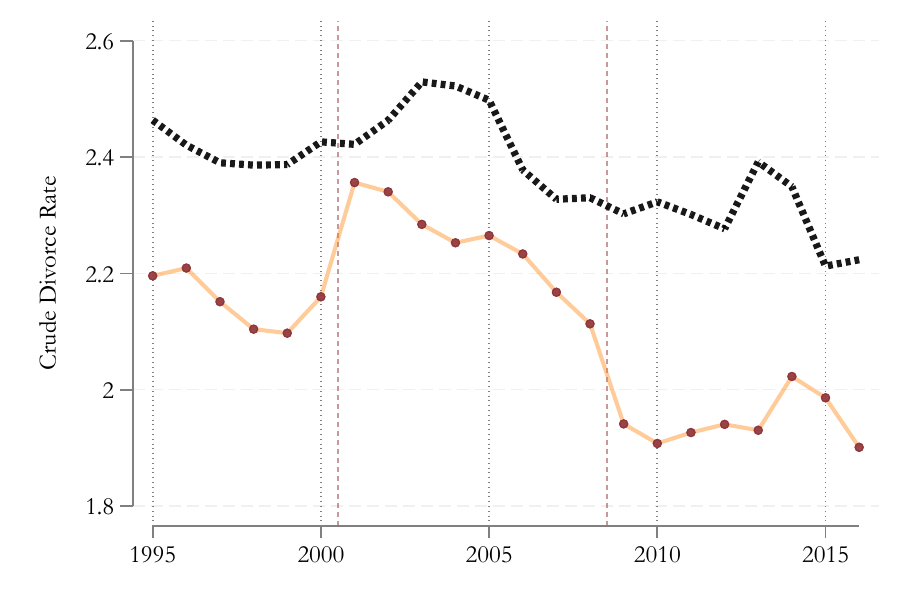}
\caption{The Netherlands vs Synthetic Control}
\label{fig_div_inc_EUR_NL}
\end{subfigure}

\begin{tablenotes}
Data from Eurostat. Top to bottom in Fig 2A (at the start of the observation period): UK, FI, SE, DK, Synthetic Control with uniform weights, NO, DE, FR. 
\end{tablenotes}
\end{figure}
 
\section{Model}
\label{Economic Theory}

To explore the marriage market dynamics and welfare changes induced by the availability of the administrative divorce option, we employ a structural model. 
At the heart of our model are individual decisions of marriage market participants. These decisions are made within the context of a broader marriage market population and, thus, constrained by scarcity (which partners are available) and subject to prices (such as the costs of divorce). The equilibrium state of this population evolves over time, reflecting endogenous changes of the marriage market (\textit{e.g.}, higher divorce rates under the low-cost regime) and broader demographic trends. 
Holding preferences fixed, this framework allows us to evaluate the adjustment path of the marriage market in response to the changing divorce costs, and calculate counterfactual welfare under alternative policy scenarios. 

\subsection{Individual Decisions}
\label{theoryindiv}
First, we introduce the individual decision-making process.
In each period, we consider a two-sided matching market where every individual can match to exactly one person from the other side of the market, or stay single.
Each match generates a surplus $ \Phi $ through joint consumption of market goods and the output of household production. 
Instead of explicitly modelling the trade-off between continuous consumption and time spent on home production, we approximate it by a discrete choice between specialisation in home production $ \sf{H} $ and market work $ \sf{W} $. 
This specialisation is contingent on being married.\footnote{This is an empirical stylised fact and can easily be generalised within our framework.}
Denoting currently married individuals as $ \sf{C} $ and unmarried individuals as $ \sf{U} $, we get a characterisation of individuals into types $ m, f \in \mathcal{I} \equiv \{\sf{U}, \sf{CH}, \sf{CW} \} $.
These types are model-determined and subject to endogenous decisions, which means that they can evolve over time.

The matching rubric is defined as follows. Unmarried individuals can either match with an unmarried partner or remain single. Married individuals can match with a married partner (of any type), or, in the event of divorce, become single.
The latter margin identifies insider marriage surplus from the matching decision against the outside option of divorce for each type. 
Divorced individuals retain their type from the previous marriage up to the point when they find a new partner.
This modelling choice allows us to distinguish re-partnering rates by marriage roles. 
In the matching rubric, we denote single households (never married and divorcees) as individuals matched with a dummy type $ 0 $. This leads to the finite characterisation of feasible matches in Table \ref{tab:feasible}.

\begin{table}[h!]
\caption{Feasible Types of Matches} 
\label{tab:feasible}
\centering
\begin{tabular}{|l|c|c|c|c|}
\hline
\diagbox{\textbf{men}}{\textbf{women}}& unmarried & married \& home prod. & married \& working & 0 \\ \hline
unmarried & \mytick  &  \mycross  & \mycross  &  \mytick  \\ \hline
married \& home prod. &  \mycross  &  \mytick  &  \mytick  &  \mytick  \\ \hline
married \& working &   \mycross &  \mytick  &  \mytick  &  \mytick  \\ \hline
0 &  \mytick  &  \mytick  &  \mytick  &   \\ \hline
\end{tabular}
\begin{tablenotes}
The feasible combination [unmarried; unmarried] denotes an unmarried cohabiting couple, [married \& ...; married \& ...] denotes a married couple with the given household specialisation type, [unmarried; 0] denotes a single individual who has never been married, and [married \& ...; 0] denotes a single divorced individual with the household specialisation type observed in the previous marriage.      
\end{tablenotes}
\end{table}

Aggregating the underlying rich household structure to this finite set of observed types introduces a systematic type-level surplus $ \Phi_{mf} $ which, in equilibrium, can be decomposed into individual utilities $ \Phi_{mf} \equiv U_{mf} + V_{mf} $.
The surpluses $ \Phi_{m0} $ and $ \Phi_{0f} $ of single men and women are normalised to zero and treated as outside options.  

Empirically, two couples with the same observed types may differ in how naturally they settle into their specialisation roles, how well they coordinate, or how productive they are in these roles. To account for such differences, we allow for idiosyncratic compatibility parameters that shift preferences over the partner's type away from what the systematic surplus alone would imply. 
The variance of these parameters is allowed to differ across types, and individuals draw their parameter realisations at the beginning of each period. 
That is, at time $ t $, a man $ i $ of type $ m $, draws a compatibility parameter for each feasible partner type $ f $, including the outside option, denoted $ \varepsilon_{ift} $ and, similarly, a woman $ j $ of type $ f $ draws $ \eta_{mjt} $ for each $ m $, both from a Gumbel distribution. 

The individual dynamics are operationalised as follows. Men and women join the matching market upon entering adulthood as unmarried types $ \sf{U} $. From this point onwards, they start making decisions whether to match with an unmarried partner or remain single.
If we were to abstract from future decisions, a man $ i $ and a woman $ j$ would match if their contemporaneous utilities, net of an equilibrium transfer $ \tau_{\mathsf{UU}} $ from $ i $ to $ j $, exceeded their respective outside options:
\begin{equation}\label{eq:qs}
  U_{\mathsf{UU}} + \sigma_\mathsf{U} (\varepsilon_{i\mathsf{U}t} - \varepsilon_{i\mathsf{0}t}) - \tau_{\mathsf{UU}}  > 0 \qquad \text{and} \qquad V_{\mathsf{UU}} + \sigma_\mathsf{U} (\eta_{j\mathsf{U}t} - \eta_{\mathsf{0}jt}) + \tau_{\mathsf{UU}} > 0 
\end{equation}
where $ \sigma_U $ are the scales of their idiosyncratic tastes to be in a relationship. %
Of course, individuals also consider their future when they choose a partner.
Denoting the inclusive values of the match in \eqref{eq:qs} as $ Q^m_{\mathsf{UU},i,t} $ and $ Q^f_{\mathsf{UU},j,t} $, respectively, and $ Q^m_{\mathsf{UU},t + 1} $ and $ Q^f_{\mathsf{UU},t+1} $ their expected values with respect to the uncertainty faced from time $ t + 1 $ onwards, under rational expectations, individuals will match if their lifetime expected surplus is positive:
$$ Q^m_{\mathsf{UU},i,t} + \beta Q^m_{\mathsf{UU},t+1} > 0 \qquad \text{and} \qquad Q^f_{\mathsf{UU},j,t} + \beta Q^f_{\mathsf{UU},t+1} > 0. $$

What makes this a truly dynamic problem is the fact that types can change.
Coupled individuals may decide to marry, which changes their match type to $ \sf{C} $.\footnote{At the individual level, partner choice and marriage are defined as sequential decisions. Marriage cannot materialise straight from two singles.}
While married, partners draw realisations of $ (\varepsilon_{i\mathsf{CH}t}, \varepsilon_{i\mathsf{CW}t}, \varepsilon_{i\mathsf{0}t}) $ and $ (\eta_{\mathsf{CH}jt}, \eta_{\mathsf{CW}jt}, \eta_{\mathsf{0}jt}) $, which determine how productive they are together in each specialisation role.
Based on this, each partner decides between one of two specialisations: home production or market work.
The household will stay married if, for some combination of roles $ mf \in \{ \mathsf{CH}, \mathsf{CW} \} \times \{ \mathsf{CH}, \mathsf{CW} \} $, there is a transfer $ \tau_{mf} $ that makes the inside value exceed the outside option for both partners:
\begin{align}
  U_{mf} + \sigma_m (\varepsilon_{ift} - \varepsilon_{i0t}) - \tau_{mf} + \kappa_{mf}\Psi_{mf}(c_d) + \beta Q^m_{mf,t+1} & > 0, \\
  V_{mf} + \sigma_f (\eta_{mjt} - \eta_{0jt}) + \tau_{mf} + (1-\kappa_{mf})\Psi_{mf}(c_d) + \beta Q^f_{mf,t+1} & > 0
\end{align}
Here, $c_d$ represents the cost of divorce. This cost depends on the policy regime $d$, with $d = 0$ denoting the high-cost regime, and $d = 1$ denoting the low-cost regime. 
Remaining married avoids the cost of divorce. We model the avoided cost as a utility gain $\Psi_{mf}(c_d)$, which we label as `stay premium'. The partners share this premium in proportions $ \kappa_{mf} $ and $ 1-\kappa_{mf} $.\footnote{The parameter $ \kappa_{mf} $ is not identified from our data, and thus, absorbed into the equilibrium transfer.}
We define each partner's inside value and continuation value of marriage analogously to those of unmarried individuals.
Because the couple's expected continuation value ($ Q_{mf,t+1} $) is the sum of both partners' values,\footnote{This follows from linearity of contemporaneous utility and the discounted continuation value. We derive this separability in detail in Appendix \ref{sec:generalmodel}.} we can write their participation in marriage as the following joint condition:
\begin{equation}
  \Phi_{mf,0} + \Psi_{mf}(c_d) + \sigma_m(\varepsilon_{ift} - \varepsilon_{i0t}) + \sigma_f (\eta_{mjt} - \eta_{0jt}) + \beta Q_{mf,t+1} > 0.
\end{equation}

The couple will choose to dissolve their union if their joint utility deficit exceeds the negative utility associated with the monetary cost for divorce.\footnote{Note that the continuation value has a stabilising effect as it balances out bad realisations of the compatibility parameters.} The divorced individuals keep their type until they re-partner, which accounts for the fact that their success on the matching market might depend on the role they had within their marriage.

Taking expectations over the idiosyncratic compatibility shocks, the marriage participation conditions induce a type-level expected surplus from a feasible match $mf$.
The divorce cost regime enters this surplus monotonically through the stay-premium $\Psi_{mf}(c_d)$.
Since we only ever observe couples under one of the two cost regimes, the level of $\Psi_{mf}$ is absorbed into the baseline surplus and only the regime
difference is identified.
We therefore write the regime-specific surplus as
$ \Phi_{mf}(d) = \Phi_{mf,0} - d\,\psi_{mf} $,
where $ \psi_{mf} \equiv \Psi_{mf}(c_0) - \Psi_{mf}(c_1) \geq 0 $ is the decline    
in the surplus of remaining married when divorce costs become low.\footnote{The surplus $ \Phi $ can, generally, depend on time $ t $ beyond the time-dependent cost regime. As an extension of our baseline empirical specification, we also estimate a specification with a linear trend balancing the trade-off between flexibility and parameter dimension.}

Aggregating individual choices to the population leads to a dynamic assignment problem. 
Given current masses of men and women describing the population, the market mechanism selects a feasible distribution of couples across type pairs. 
Through match-type-dependent within-household decisions such as marriage, specialisation, and role-switching this determines how the current distribution of couples translates into the next period's distribution of individual types. 
We formalise this in the next subsection.

\subsection{The Population Perspective}
\label{subs:poppersp}

Let $M_{mt}$ and $F_{ft}$ denote the masses (population proportions) of men and women of type $m,f\in\mathcal I$ in period $t$, and let $\mu_{mft}$ denote the mass matched in pair $mf$ (including the outside options). 
A feasible allocation $\mu_t$ must exhaust the masses on both sides of the market and assign zero mass to infeasible matches specified in Table \ref{tab:feasible}.

The key observation is that the individual decision problem aggregates to a population Bellman equation (see Appendix \ref{sec:generalmodel}). 
For a given population state $(M_t,F_t)$, the period allocation $\mu_t$ determines both current welfare and next period's type masses. 
From a social planner's perspective, we can write contemporaneous welfare as the sum of two deterministic terms: 
\begin{equation}
  \label{eq:bellman_t1}
r(M_t,F_t,\mu_t;d) = \sum_{m,f}\mu_{mft}\Phi_{mf}(d) + \mathcal E_\sigma(\mu_t,M_t,F_t).
\end{equation}
That is, the population-weighted sum of match-specific mean surpluses of each feasible couple type, and the aggregate utility  $\mathcal E_\sigma $ that the population derives from the additive, random preferences unobserved to the econometrician \citep{GalichonSalanie2022}. 
Rational expectations of individuals allow us to write aggregate welfare as a deterministic population-level quantity, which translates to perfect foresight from the social planner's perspective.

Transitions from current matches to next-period individual types are summarised by Markovian transition kernels, which we identify non-parametrically:
\begin{equation}\label{eq:lawofmotion}
M_{t+1}=\mathcal P^m_d(\mu_t), \qquad F_{t+1}=\mathcal P^f_d(\mu_t).
\end{equation}
Both the matching stage and the kernels depend on the divorce-cost regime $d$. This allows us to separate the effects on divorce rates and marriage rates. In the matching stage, divorce costs enter the surplus $ \Phi_{mf}(d) $ and a couple dissolves by choosing singlehood over remaining matched, so divorce is a decision of insiders, taken against the outside option. In the kernels, we allow for a regime-specific inflow of outsiders into marriage and selection of roles, which governs the effect on marriage rates. 

Marriage rates may rise when divorce is cheaper, as the lower expected contractual cost makes entry more attractive to outsiders. The same fall in cost, however, lowers the matched surplus $\Phi_{mf}(d)$ and so raises divorce rates among insiders. The married stock can therefore shrink or expand, depending on which of the two margins dominates.\footnote{The two margins are separated by design. Kernels cannot map a married individual into singlehood, because divorce is a matching decision.}
The population value function then becomes:
\begin{equation}
  \label{eq:bellman_t2}
Q_d(M_t,F_t) = \max_{\mu_t\in\Gamma(M_t,F_t)} \left\{ r(M_t,F_t,\mu_t;d) + \beta Q_d \left(\mathcal P^m_d(\mu_t),\mathcal P^f_d(\mu_t)\right) \right\}, 
\end{equation}
where $\Gamma(M_t,F_t)$ is the set of feasible couplings that allocates all individuals into feasible couple types defined in Table \ref{tab:feasible}. The maximised term is the sum of the population's contemporaneous welfare and the discounted continuation value.
This recursive problem is the planner representation of the marriage market: they choose the type-level assignment that balances current matching surplus against its effect on future type distributions.

\citet*{Corblet2023} show that, in the absence of blocking pairs \citep{Shapley1971}, the solution to this planner problem coincides with the decentralised market equilibrium. 
Appendix \ref{sec:generalmodel}, and Assumptions \ref{ass:TU}-\ref{ass:transprob} therein, formalises this equivalence using the type-based random utility structure and the additive structure of the entropy-regularised optimal transport representation of equation \eqref{eq:bellman_t1}, which is preserved in the Bellman equation \eqref{eq:bellman_t2}.
This characterisation is important for two reasons. First, it tells us that the market outcome is efficient. Second, it allows us to identify individual utilities and match surpluses by inverting the planner's problem, rather than having to go through a rationalisability argument based on decentralised individual-level choices.

This model representation of the Dutch marriage market can be empirically estimated thanks to the rich structure of our administrative data. We observe the maximiser of \eqref{eq:bellman_t2} as the whole path of equilibrium matches $ \{\mu_t\}_{t=t_0}^T $ over a long time window $t \in \{t_0, T\} $.
The surplus matrix $ \Phi $ is identified through fitting the model's implied matching allocation trajectory to the observed one.
Simulating population states $ (M_t, F_t) $ that differ from the path observed in our data allows us to calculate the counterfactual welfare trajectories.\footnote{We discuss the empirical specification and the objective function for estimation in Section \ref{sub:parest} and provide comprehensive details about computation in Appendix \ref{app:computation}.}
Further, the longitudinal nature of our individual-level records enables us to directly identify the laws of motion. This is done through the year-to-year transition kernels $ \mathcal{P}^m_d $ and $ \mathcal{P}^f_d $ in equation \eqref{eq:lawofmotion} which determine marriage, specialisation, and role switching.
Finally, through Lemma \ref{lem:hazard}, the discontinuity in divorce rates around the 2009 reform identifies the marriage surplus difference between the two divorce-regimes, denoted $ \psi $.
We discuss our data in the next section.

\section{Data}
\label{Data and Descriptive Statistics}

\subsection{Data sources and sample construction}

The empirical analysis rests on a dataset constructed from several administrative registers maintained by Statistics Netherlands (CBS). The cornerstone of the dataset is the municipal register (\textit{Gemeente Basis Administratie}), which covers all Dutch residents from 1995 onward: their date of birth, gender, immigration background, full marital history with spousal identifiers, full fertility history with child identifiers, and extensive residential history with household characteristics, within-household roles: household head, spouse, child, or unspecified, each with personal identifiers.\footnote{If there are several couples living on the same property (e.g., three-generation families or communal house-sharing arrangements), each couple is treated as a separate household unit.}

We consolidate this information into a longitudinal dataset that tracks the marriage market outcomes of the Dutch population. In line with our matching rubric (Table \ref{tab:feasible}), we distinguish three relationship types: single, cohabiting with a partner, and married.\footnote{Couples in registered partnerships are classified as cohabiters. This is motivated by their relative scarcity in the data, and by the fact that their unions were not affected by the availability of flash divorce. For details, see Appendix \ref{sec:regpar}.} To account for marital instability, single individuals are further split into those who never married and those who divorced. Leveraging deterministic linkages to the employment register, we retrieve the labour market status of all married individuals. Those working full-time are classified as specialising in market work, whereas those working part-time (or not at all) are classified as specialising in home production. This information, along with spousal identifiers, enables us to assign each individual into the correct feasible match type $mf$, and determine the equilibrium assignment. The most common match types are: married couples with male breadwinner, married couples with both spouses working, cohabiting couples, and singles. 

\begin{figure}[ht!]
  \caption{Population snapshot of relationship status and marital status by age}
    \centering
    \begin{subfigure}[b]{0.65\textwidth}
        \centering
        \caption{Households by relationship type} 
        \label{Relationship status} 
        \includegraphics[width=\textwidth]{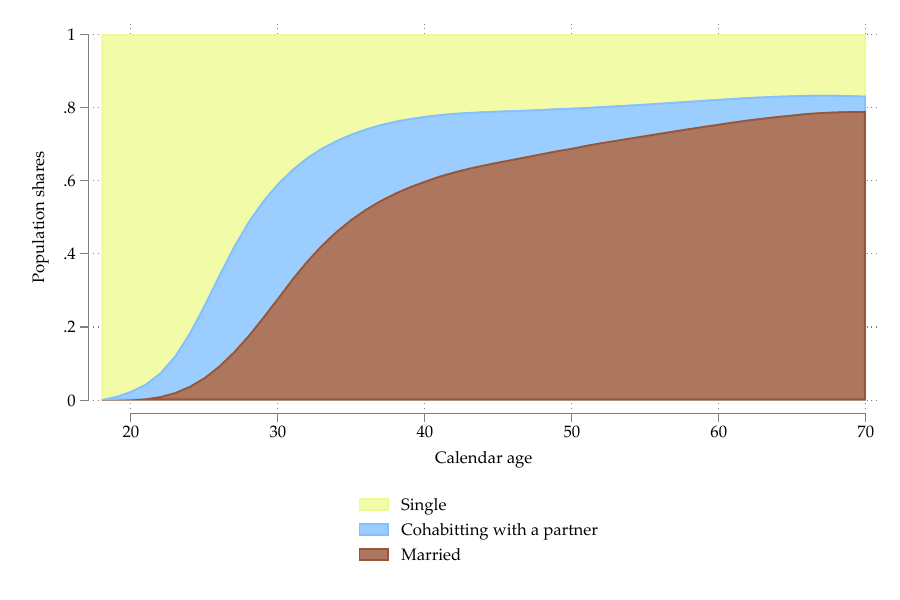} 
    \end{subfigure}
    \vspace{0.5cm} 
    \begin{subfigure}[b]{0.65\textwidth}
        \centering
        \caption{Individuals by marital status} 
        \label{Marital status} 
        \includegraphics[width=\textwidth]{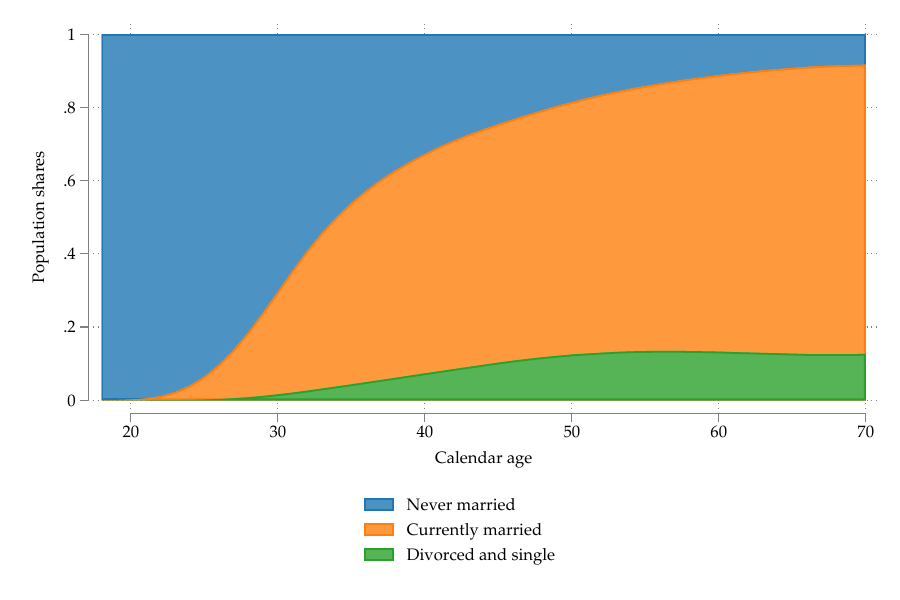}   
    \end{subfigure}
    \label{fig_dynamics}
    
\vspace{-0.8cm}
\end{figure} 

Although our model is not one of a life-cycle, it is worth taking a look at a snapshot of the population of households (Figure \ref{Relationship status}) and individuals (Figure \ref{Marital status}) at different life stages.
Most people start out in single households. By the age of 30 the majority have formed a couple, with cohabitation being the dominant mode. 
The steepest increase in marriages is in the early 30s, stabilising by 40 and growing at a constant rate from then onwards, mostly replacing cohabitation.
Roughly a fifth of the population is single at any given point.
We can also see a constant increase in the divorced population. By the age of 50, more than 10\% of individuals are unmarried divorcees.\footnote{For the sake of parsimony, this marriage market category also includes widows and widowers; their numbers in the considered age range are however negligible.}

\subsection{Sample selection criteria}

The key sample restriction is that we limit the estimation period to 2001--2020. The lower bound reflects a limitation of the employment register, which begins distinguishing part-time from full-time work only in 2001. Because we cannot determine household specialisation patterns in the years preceding the introduction of administrative divorce, we focus on the marriage-market dynamics during the period in which flash divorce was available (2001--2009) and the period following its ban, up to year 2020. The upper bound is motivated by the COVID-19 pandemic, which had a strong and idiosyncratic effect on the functioning of the Dutch marriage market.      

Additional sample restrictions are kept to the minimum. First, we do not consider the marital outcomes of same-sex couples. This is because the introduction of same-sex marriage rights coincided with the start of the administrative divorce regime. Second, we do not consider the marital outcomes of people who are institutionalised (mainly prisoners and clients of mental health facilities), since these individuals are not participating in the marriage markets. Fourth, we apply modest age restrictions. We do not consider the outcomes of people who are younger than 18, nor the outcomes of the elderly, since the data show that individuals who remain married past their 60s are unlikely to ever divorce (see Appendix Figure \ref{fig_age_dif}). 
We apply different age thresholds for men (61) and women (57), corresponding to the respective 95th quantiles of the divorcee's age distribution.

\subsection{Separation, re-partnering, and demographics}
 
The higher divorce rate during the low-cost regime (as shown in Figure \ref{fig_div_inc_NL}) is accompanied by changes on other margins.
First, Figure \ref{sep_div_days} reports the average number of days between residential separation, a proxy for the relationship breakdown, and the day of divorce.
Consistent with procedural simplicity being one of the appealing features of the administrative option, divorces were resolved more quickly during the period in which flash divorce was available.

\begin{figure}[h!]
\caption{Average number of days between residential separation and divorce }\label{sep_div_days} 
\includegraphics[trim={0 0 0 0},clip,width=\textwidth]{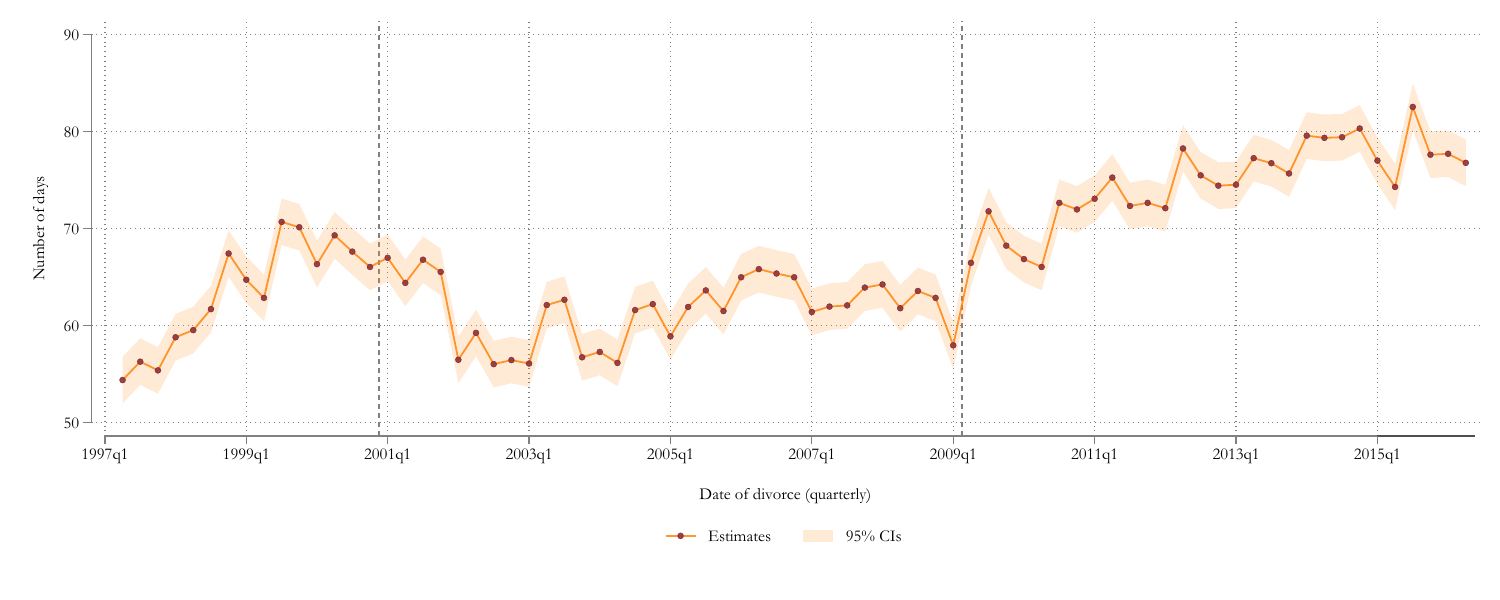} 
\vspace{-0.8cm}
\end{figure}

\begin{figure}[htbp] 
     \caption{Incidence of re-partnering among Dutch divorcees}\label{repar_inc} 
    
    \begin{subfigure}{\textwidth}
         \caption{Men}\label{repar_inc_men} 
         \includegraphics[width=\textwidth]{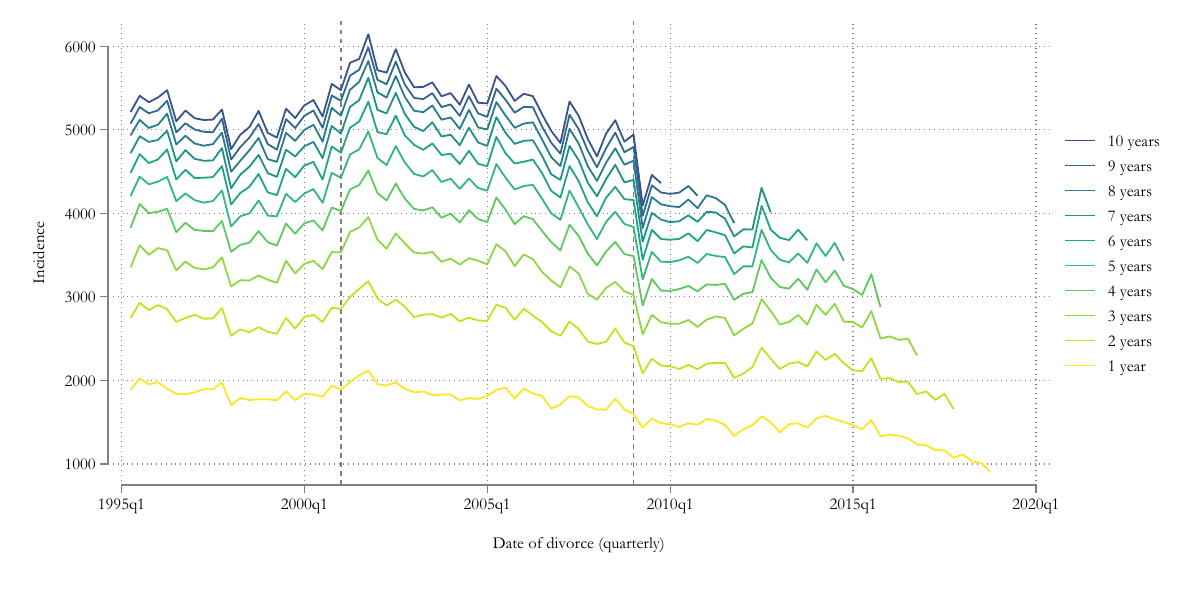}  
    \end{subfigure}
 
    \begin{subfigure}{\textwidth}
        \caption{Women}\label{repar_inc_women}  
         \includegraphics[width=\textwidth]{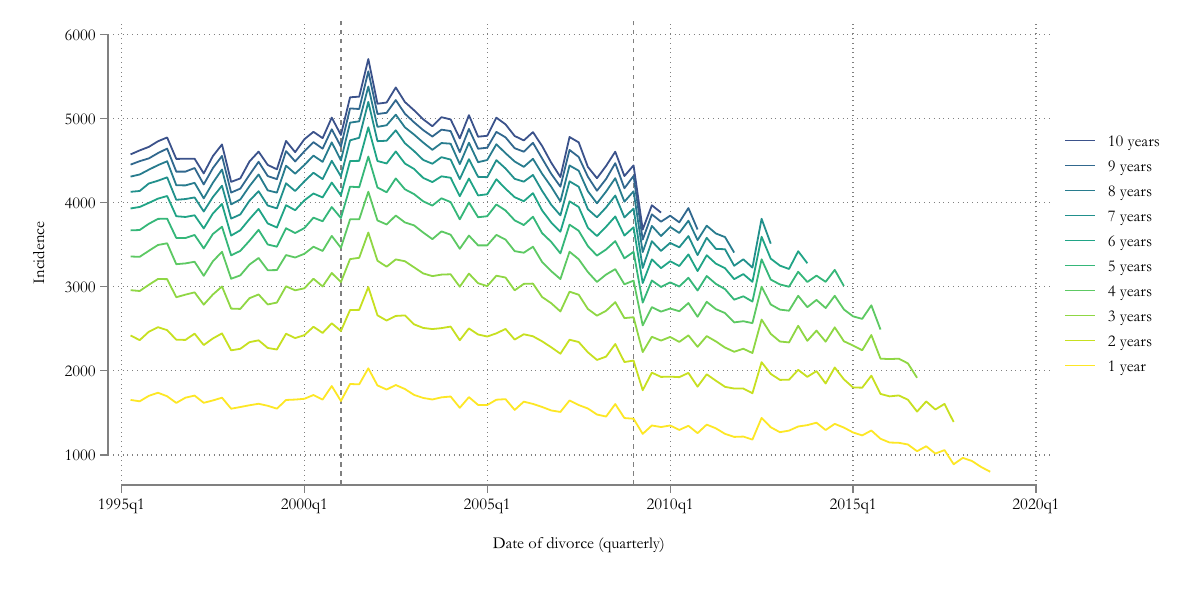}  
    \end{subfigure}
\begin{tablenotes}
Statistics Netherlands Microdata 1995-2020. Re-partnering incidences are plotted across time, quantified at 1--10 years past the date of divorce. The observability of re-partnering incidences at longer time horizons is constrained by the limits of our observation period.  
\end{tablenotes}

\end{figure}

Second, Figures \ref{repar_inc_men} and \ref{repar_inc_women} show the incidences of re-partnering among men and women who divorced within observation window. To ensure comparability across time, we plot the incidences at different time horizons (from 1 to 10 years past the date of divorce). We see that the flash divorce period was characterised by a higher incidence of re-partnering, which is consistent with the higher observed inflows of divorcees back into the marriage market. Remarkably, the trajectories for longer time horizons show that the incidence of re-partnering was increasing even among men and women who divorced prior to 2001. This suggests that earlier divorce cohorts also benefited from being able to encounter more potential partners during the flash-divorce period. Another notable observation is that the long-term incidence of re-partnering falls over the flash divorce period, and plateaus after the 2009 ban. This aligns with the fact that, because of the ban, the later cohorts of divorcees were interacting with a gradually shrinking pool of potential partners. Past the 2009 ban, the inflow of new divorcees reverted to the lower level, and the size of the pool stabilised. 

These narratives capture a central feature of the policy variation we study: although the administrative option increases divorce incidence immediately, its effects on broader marriage markets unfold over much longer horizons. The equilibrium trajectory adjusts to the altered flows of marriage market participants only gradually, which renders their effects on the match type distribution highly dynamic. To pin down the total effects of candidate policy regimes, we therefore focus on long-run counterfactual predictions that fully internalise these slower adjustments.  

In addition to the visualisation exercises above, we conduct an auxiliary logistic regression analysis to examine the characteristics of couples who opt for the administrative divorce pathway. Using the population of couples who divorced during the period of flash divorce availability, the regression model associates its take-up with a rich set of socio-economic controls (see Appendix Table \ref{tab:flashreg}). The take-up is found to be higher among native couples with similar earnings and education levels, and lower among sole-breadwinner households, immigrant couples, and couples with large income differentials.
The take-up is also decreasing with marriage duration: longer marriages, presumably involving more assets and children, are less likely to use the administrative option.
Together, we interpret this as evidence that administrative divorce is most practical for couples whose separations do not involve complex alimony or property-division negotiations.

\section{Results}
\label{Results}
\subsection{Divorce Hazards}

We begin our empirical analysis by estimating the values of $\psi_{mf}$ for each currently-married couple type  $mf\in\{\sf CH,\sf CW\}\times\{\sf CH,\sf CW\}$.
As discussed in Section \ref{theoryindiv}, these are the type-specific declines in the marital surplus for insiders caused by moving from the high-cost to the low-cost divorce regime.
To identify this, we exploit the structure of our model. Due to the extreme value assumption on the compatibility parameters, the matching problem has a two-sided logit structure.
This identifies the mean surplus $\Phi_{mf} $ for each currently-married couple type, and the scales $ \sigma_m$ and $\sigma_f $ of the compatibility parameters' distributions, up to a normalisation.
Let $\lambda_{mf}(d)$ be the per-period divorce hazard under divorce-cost  regime  $ d $, \textit{i.e.}, the probability that the match dissolves between periods $t$ and $t+1$ conditional on being in state $mf$ at $t$. 
Under our random-utility specification, changes in divorce costs translate into shifts in these hazards.
In Lemma \ref{lem:hazard} in Appendix \ref{app:hazard}, we show that the low-cost regime lowers the net surplus from remaining married for type $mf$ approximately by an additive log-hazard ratio:
\begin{equation}
  \psi_{mf} \approx \bar{\sigma}_{mf}  \log \left(\frac{\lambda_{1,mf}}{\lambda_{0,mf}}\right),
\end{equation}
where $ \bar{\sigma}_{mf} = \sqrt{\sigma_{m}^2 + \sigma_{f}^2} $ is the marriage-type-specific scale parameter.

We calculate these hazard rates for each couple type based on a regression discontinuity design (RDD). Centering around the event date of the ban of the flash divorce in 2009, we estimate the drop of the type-specific divorce rate at the cut-off.\footnote{In principle, a simplified version of our model (that abstracts from household specialisation choices) could leverage the introduction of the flash divorce in 2001 for a test of over-identifying restrictions and identification of asymmetric effects.}
To illustrate the design and its validity, Figure \ref{RD} reports an unconditional RDD chart for all married couples (pooling across specialisation types), and compares it to the RDD chart for cohabiting couples.  
The first chart shows that the divorce rate experienced a significant drop (11.5\%) at the point when the flash divorce option was banned, whereas the separation rate of cohabiting couples remained identical unaffected by the ban. This confirms our prior that the higher incidence of divorces during the flash divorce period was not coincidental.
We note that the hazards imply an overall utility decline of marriage of $ \psi_{\mathsf{CC}} \approx 0.05 $.
We repeat this exercise to obtain $ \lambda_{mf}(d) $ for each married type $ mf $ and calculate the corresponding surplus differences $ \psi_{mf} $.

\begin{figure} [htbp]
\caption{Regression Discontinuity analysis of the flash divorce ban}
    \label{RD}
    \centering
    \begin{subfigure}{0.48\linewidth}
        \centering
        \caption{Divorce rate of married couples}
        \includegraphics[width=\linewidth]{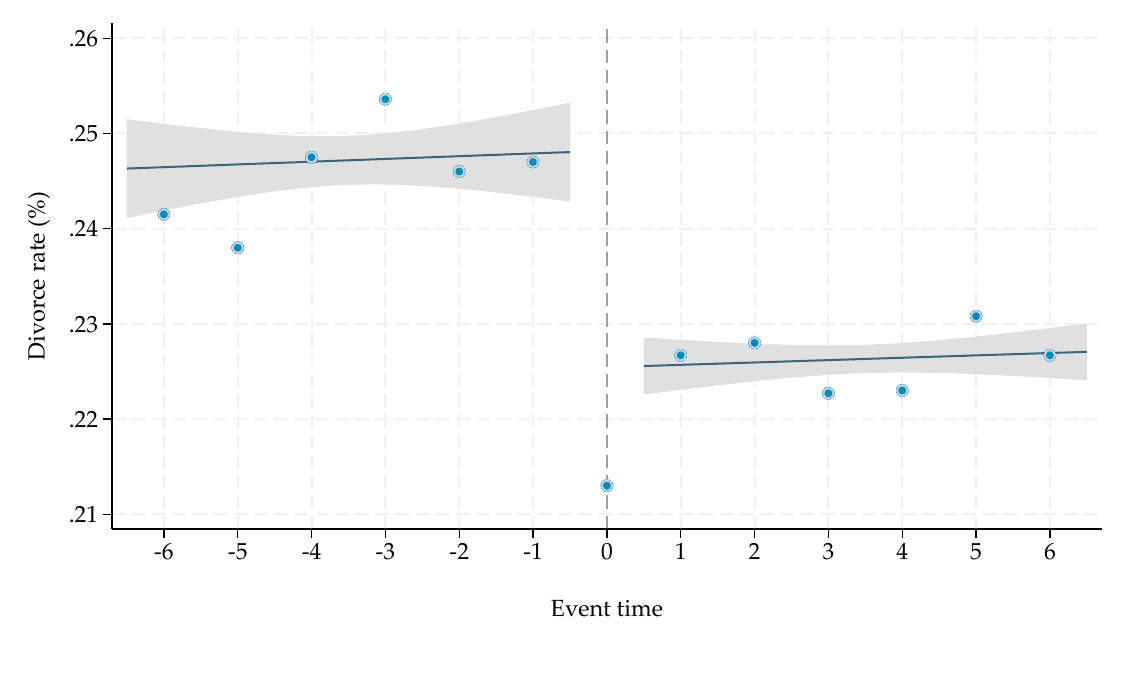}
    \end{subfigure}
    \hfill
    \begin{subfigure}{0.48\linewidth}
        \centering
        \caption{Separation rate of cohabiting couples}
        \includegraphics[width=\linewidth]{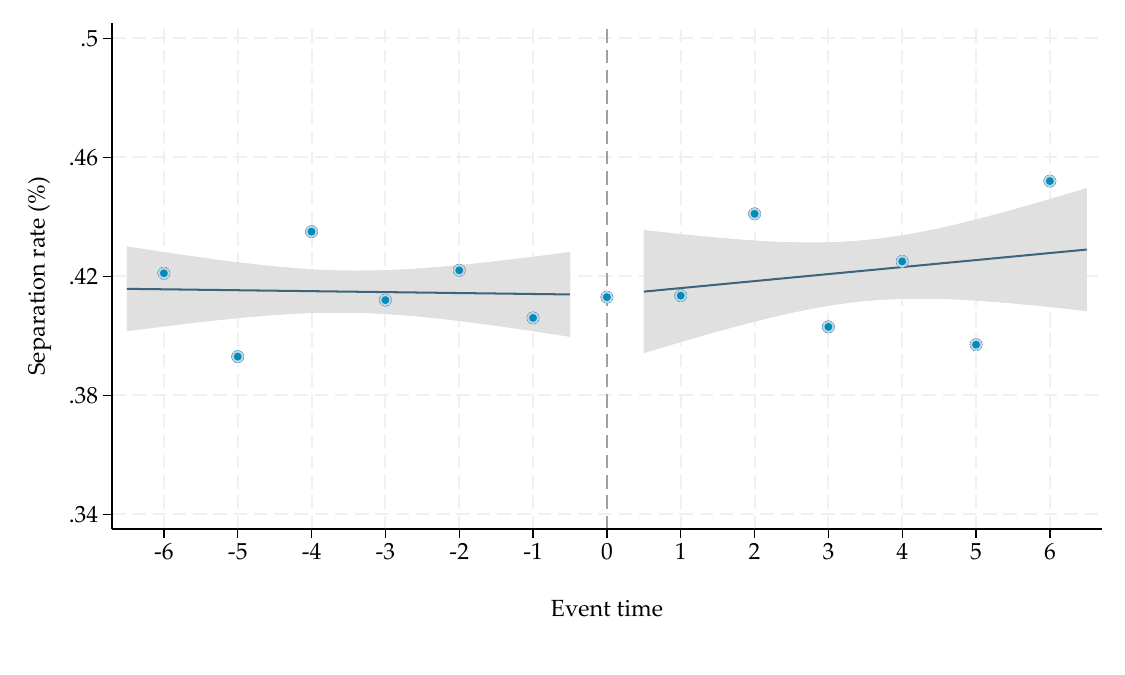}
    \end{subfigure}
    \begin{tablenotes}
      Donut regression discontinuity design centered around 2009 ban of the flash divorce option. The shaded areas represent 95\% confidence intervals around the estimates.
    \end{tablenotes}
    
\end{figure}

\subsection{Parameter Estimation}\label{sub:parest}
Our model captures two margins through which the administrative divorce option may influence marriage.
The first operates through match formation and dissolution on the matching market.
The second operates within the household: marriage and specialisation. These two margins motivate the set of match types in our model, with the matrix of all feasible matches (Table \ref{tab:feasible}) being the planner's decision space.
The state-space is the vector of male and female population shares (masses) corresponding to each individual type $ m, f \in \mathcal{I} \equiv \{\sf{U}, \sf{CH}, \sf{CW} \} $. 

To characterise population welfare, we collect the parameters which fully saturate all cells in the surplus matrix corresponding to feasible matches, and stack them into a vector $ \theta $ together with the vector of scales.
This parameterisation of equations \eqref{eq:bellman_t1} and \eqref{eq:bellman_t2} 
affects mean utilities through  $ {\Phi}_{\theta} $ and the dispersion of tastes and compatibilities through the generalised entropy $ \mathcal{E}_\theta $.

The evolution of types is governed by a transition kernel which we estimate non-parametrically.
Types change because couples make collective decisions: they marry, specialise, or switch roles.
Further, in each period there is exogenous entry and exit due to cohort flows (birth and death) and migration.
We summarise overall type transitions by the kernels $ \mathcal{P}_d^m $ and $ \mathcal{P}_d^f $ for $ d \in \{ 0, 1\} $ governing the law-of-motion in equation \eqref{eq:lawofmotion} within and outside the low administrative divorce cost regime, respectively.
We obtain their estimates by time-averaging year-to-year transition kernels which are directly obtained from the population data.\footnote{An alternative approach would be to take the exact transition kernels for every year, or any function thereof, such as a moving average. However, because we want to calculate long-term counterfactual forecasts, we restrict ourselves to constant transition kernels within each regime.}

This is an inverse structural problem: we observe a trajectory of equilibrium distribution of matches $ \{ \hat{\mu}_t \}_{t=t_0}^T $ and seek to recover the parameter vector $\theta$ that rationalises it. For any fixed $\theta$, the parameterised surplus matrix $\Phi_\theta$ and entropy term $\mathcal E_\theta$, together with the estimated transition kernels, pin down a dynamic equilibrium through the Bellman equation and the law of motion. Solving the model therefore yields the trajectory of equilibrium coupling matrix $ \{ \mu_t^*(\theta) \}_{t=t_0}^T$ implied by that candidate parameter.
Because the state space is high-dimensional, with six continuous state variables capturing the masses of individual types for both sides of the matching market, we show how to approximate the Bellman value function and the corresponding optimal policy rule with a neural network \citep{maliar2021} while imposing the equilibrium constraints. 
We achieve this by introducing an explicit \citet{sinkhorn1964} scaling layer in the neural network (see Appendix \ref{app:computation}).
We then estimate $\theta$ by maximum simulated likelihood, implemented as the minimisation of the Kullback--Leibler distance between observed and model-implied conditional matching distributions. In particular, we let $\mu_{mf|m,t} \equiv {\frac{{\mu}_{mft}}{M_t}}$ and $\mu_{mf|f,t} \equiv \frac{\mu_{mft}}{F_t} $ and compare $\hat{\mu}_{mf|m,t} $ and $\hat{\mu}_{mf|f,t} $ to their model-implied counterparts, which makes the criterion less sensitive to estimation error in the transition kernels governing the marginal type distributions.
We regularise the estimation by adding a penalty term to the criterion, which punishes deviations from the structural model \eqref{eq:bellman_t1}-\eqref{eq:bellman_t2} and thus prevents over-fitting to the observed matching distribution.
We discuss estimation in detail in Appendix \ref{app:estimation}.

\subsection{Results}
With the model set up and estimated, we turn to the results. We first discuss the estimates of surpluses for each feasible match. %
We start from a simplified model that only distinguishes between two types of men and women: Unmarried and Married.\footnote{All model specifications use a discount factor of $ \beta = 0.9 $, equal weighting $ w_1 = w_2 = 1.0 $ of the Bellman residual and the first order conditions, and $ \alpha = 2.0 $ for the regularisation penalty. Our neural nets have two hidden layers with $128$ and $ 64 $ nodes, respectively. Entropies of the undefined cells are excluded from contemporaneous surplus.}
The match of two unmarried individuals constitutes a cohabiting relationship, the match of two married individuals constitutes a marriage, and the outside option is being unmarried and single. Table \ref{Table_2x2} lists the surpluses associated with the two relationship types, showing that couples derive greater utility surplus from being married (3.31) than cohabiting (-1.12). 
Since the latter is often a transitory state, we explain the negative contemporaneous utility of being in a cohabiting relationship by a high enough expectation about the continuation value of marriage in the future.
The off-diagonal elements are undefined, because an (un)married man is, by definition matched with an (un)married woman, which makes this equivalent to a pair of two-sided aggregate logit choice models with the transition kernel moving masses between them. 

The surplus associated with marriage incorporates a utility decline of $ \psi_{\mathsf{CC}} \approx 0.05 $ when $ d = 1 $.
This means that the option of cheap exit reduces the value of marriage for insiders, lowering its surplus even for married couples who never divorce.

\begin{table}[H] 
\centering
\caption{Surplus Matrix $\Phi(d)$: Minimal Specification}
\begin{tabular}{@{}c lcc@{}}
\label{Table_2x2}
 &  & \multicolumn{2}{c}{Women} \\ 
\cline{3-4}
 &  & unmarried & married \\
\cline{2-4}
\multirow{2}{*}{\rotatebox[origin=c]{90}{Men}} 
 & unmarried & $-1.12$ & $\cdot$ \\[1ex]
 & married & $\cdot$ & $3.31 - 0.05 d$ \\[1ex]
\cline{2-4}
\end{tabular}
\begin{tablenotes}
  The variable $ d = 1 $ indicates the low-cost regime, the coefficient on $ d $ is the estimated surplus decline $ \psi_{mf} $. 
Outside options of singlehood are normalised to zero and not shown in the table.
\end{tablenotes}
\end{table}

To investigate the consequences of higher divorce costs on household specialisation, we now distinguish married individuals engaged in market work ($W$) and those engaged in home production ($H$). The surplus matrix corresponding to our principal model specification is presented in Table \ref{Table_3x3}. 

\begin{table}[H]
\centering 
\caption{Surplus Matrix $\Phi(d)$: Baseline Specification}
\begin{tabular}{@{}c lccc@{}} 
\label{Table_3x3}
 &  & \multicolumn{3}{c}{Women} \\
 \cline{3-5}
 &  & unmarried & married, home prod. & married, work \\
\cline{2-5}
\multirow{3}{*}{\rotatebox[origin=c]{90}{Men}} 
 & unmarried & $-1.13$ & $\cdot$ & $\cdot$ \\[1ex]
 & married, home prod. & $\cdot$ & $1.64 - 0.12 d$ & $0.55 - 0.00 d$ \\[1ex]
 & married, work & $\cdot$ & $1.47 - 0.00 d$ & $1.90 - 0.06 d$ \\[1ex]
\cline{2-5}
\end{tabular}
\begin{tablenotes}
  The variable $ d = 1 $ indicates the low-cost regime, the coefficient on $ d $ is the estimated surplus decline $ \psi_{mf} $. 
Scales: $ \log\sigma = (0, -1.27, -1.36)$, loss: $ D_{\text{KL}} = 0.21 $. 
\end{tablenotes}
\end{table}

As before, individuals derive greater utility surplus from being married than from being unmarried. The utility surpluses differ among the four housework and market work configurations, being the highest for married couples in which both spouses are focused on market work (1.90), and lowest for married couples in which the woman is engaged in market work and the man in housework (0.55). 
These marriages seem to be the most unstable, leading to the highest number of divorces.
We further see that the low divorce cost regime decreases the utility surplus associated with marriages in which both spouses focus either on market work (-0.06) or housework (-0.12). 
In contrast, the low divorce cost regime does not lower the utility surplus of single-earner couples (-0.00). This apparent stability is in accordance with economic intuition: single-earner couples are unlikely to take advantage of administrative divorce, as their separations are complicated by alimony and property division negotiations, which would be difficult to settle and enforce out of court.\footnote{Appendix Table \ref{Table_3x3t} reports structural estimates with time-trends in surpluses; the results are qualitatively similar to the baseline model; there is a positive trend in the surplus for cohabitation, consistent with the increasing popularity of this type of relationship.}

\subsection{Counterfactuals}
To contextualise the practical importance of the administrative divorce option for the functioning of Dutch marriage markets, we conduct counterfactual simulations. Using our principal model specification (Table \ref{Table_3x3}), we study the marriage market dynamics under two scenarios: a scenario that involves the 2009 ban of the administrative divorce option (reverting back to the high-cost regime), and a scenario in which the administrative option remains available (retaining the low-cost regime). 

Figure \ref{Fig_CF_joint_lt} shows the outcomes of these simulations, plotting the model-implied trajectories of the population shares of match-types $ \mu_{mft} $ under the two scenarios. The red lines denote the population shares of match types that correspond to the high cost scenario, whereas the green lines correspond to the low-cost scenario. The dashed segments of the two lines denote out-of-sample predictions that cover a 20-year period (up to year 2040). As discussed in Section \ref{subs:poppersp} and equation \eqref{eq:bellman_t2} therein, market clearing requires the entire mass of men $ M_t $ and women $ F_t $ to be allocated into feasible couple types $ \Gamma(M_t, F_t) $. This means that the rows of $ \mu_t $, with schematic according to Table \ref{tab:feasible}, must add up to the proportions of individual types $ M_t $ for men, and the columns add up to the corresponding vector $ F_t $ for women. Their components add up to the population shares of men and women. Thus, when interpreting the graph, we must keep in mind that the population shares add up to one \emph{only} when weighted by the number of individuals in the respective household type.%

\begin{figure}[ht!]
  \caption{Counterfactuals: Joint Distribution of Match Types}
\includegraphics[width=0.95\textwidth]{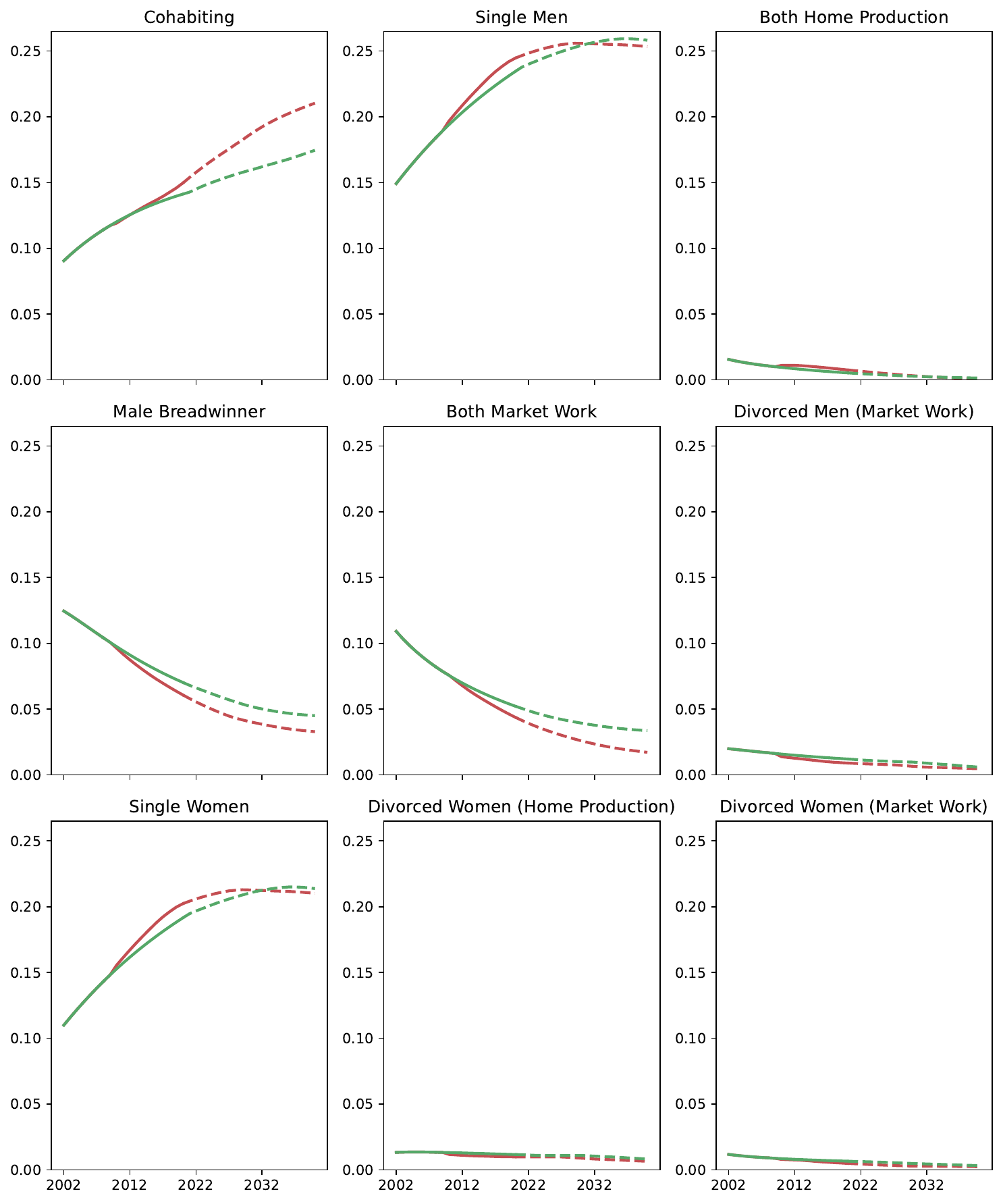}
  \label{Fig_CF_joint_lt}
\begin{tablenotes}
  Predicted population shares of joint match types under the two counterfactual scenarios. The red line corresponds to the 2009 administrative divorce ban scenario, whereas the green line corresponds to the scenario in which administrative divorce remains available. Dashed lines correspond to predictions of marriage market dynamics beyond the span of our data. Match types CH,CW and CH,0 are not featured in the plot because their population shares are close to zero. 
\end{tablenotes}
\end{figure}
The plotted trajectories of match types are characterised by time trends that reflect changing attitudes to marriage, cohabitation, and the division of tasks in the household. Specifically, the share of cohabiting couples is steadily growing (along with the shares of single men and women), whereas the shares of married couples are falling. This is well-aligned with contemporary scholarship, which emphasises the gradual de-institutionalisation of marriage \citep{Cherlin2004,Stevenson2007}. 

As illustrated by Figure \ref{repar_inc}, the reform changes the marriage market composition only gradually, reflecting the differential flow rates of marriage market participants under the two divorce cost regimes. This puts the market in a transient state for multiple years past the date of the reform. By comparing the endpoints of the two trajectories, we can evaluate the long-run effects of the high- and low-cost regimes. 

Starting with cohabiting couples, their population share increases by 3.6 p.p. (20.5\%) compared to the low-cost counterfactual. This is compensated by a fall in the population share of married couples. The largest decline is experienced by married couples in which both spouses are engaged in market work, whose share falls by 1.6 p.p. (49.5\%). This is in line with our expectations, since the spouses who are both generating income are likely to derive the most utility from the low-cost administrative divorce option. In its absence, it can be expected that many such couples would opt for cohabitation instead of marriage. We also see a 1.2 p.p. (27.1\%) decline in the share of married couples with a male breadwinner. This indicates that the high-cost regime does not encourage more couples to specialise. The decline is lower than the one observed among couples engaged in market work. We attribute this to their need to settle alimony and custody through the court, which may discourage them from taking advantage of the administrative divorce option. 

The other match types (singles, divorcees, and married couples engaged in home production) do not exhibit substantive differences in the population shares attained under the two counterfactual scenarios. The single men and women are subject to distinct dynamic adjustments, with the share of singles being initially greater under the high-cost scenario, but eventually becoming slightly lower than the share of singles attained under the low-cost scenario. This non-monotonic adjustment path illustrates the value of the long-term exposition of the counterfactual marriage market dynamics. 

\begin{figure}[h!]
  \caption{Counterfactuals: Distributions of Individual Types}
\includegraphics[width=0.95\textwidth]{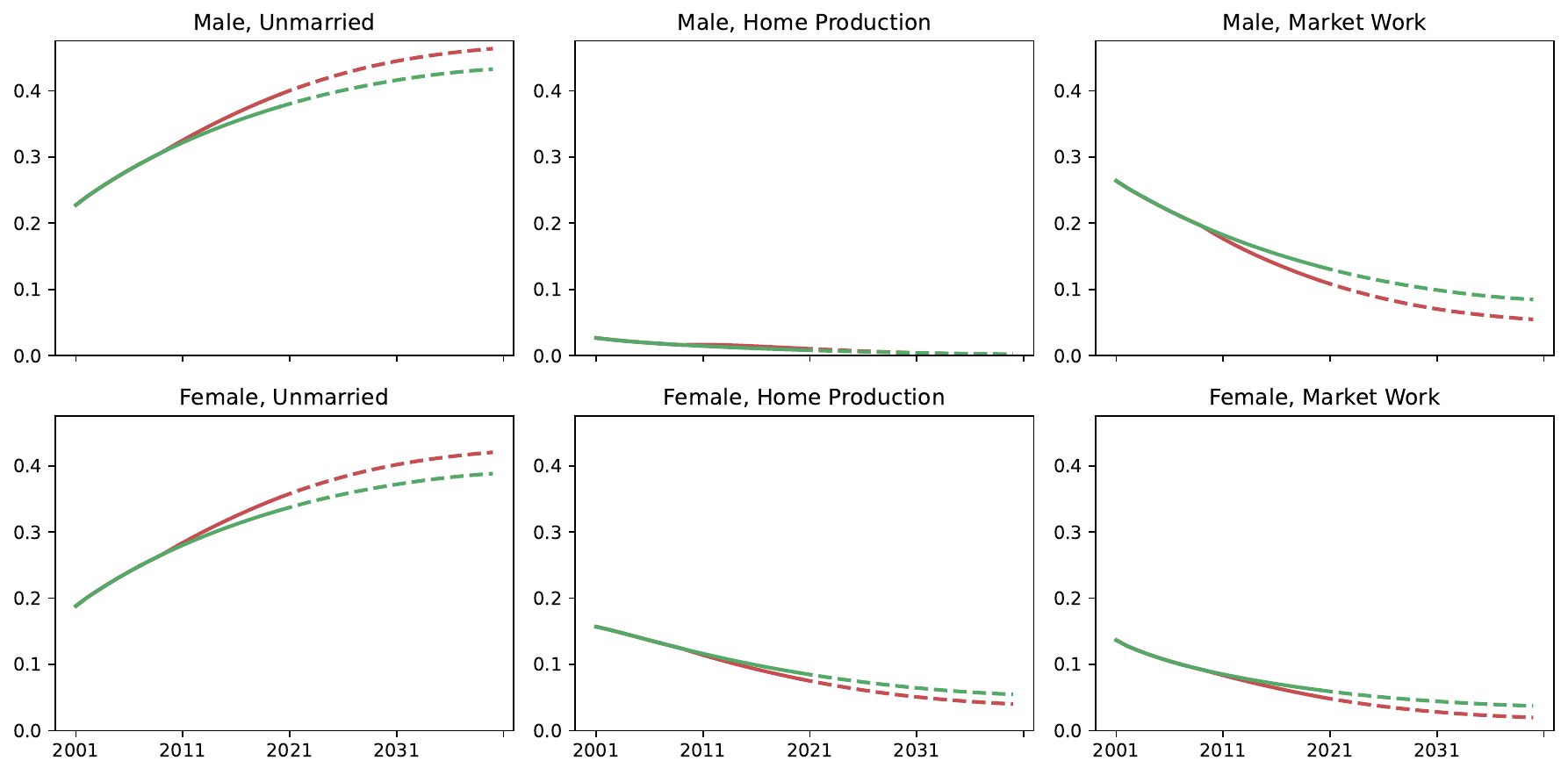}
  \label{Fig_CF_unc}
  \begin{tablenotes}
    The red line corresponds to the 2009 administrative divorce ban scenario, whereas the green line corresponds to the scenario in which administrative divorce remains available. Dashed lines correspond to predictions of marriage market dynamics beyond the span of our data.
  \end{tablenotes}
\end{figure}

Figure \ref{Fig_CF_unc} documents the evolution of individual types under the two scenarios. Under the low-cost scenario, men and women are $55.6\%$ (3.1 p.p.) and $53.7\%$ (3.2 p.p.) more likely to be married than under the high-cost scenario.
This underscores that lower divorce costs, mediated through the administrative divorce option, can indeed have a substantive dampening effect on the overall decline of marriage rates.

How does this translate into specialisation?
Of the 3.1 p.p. higher number of married women in the low cost regime, about half (1.5 p.p.) specialise in home production by the end of our prediction window, a 36\% increase compared to the high cost regime.
Almost all of the extra married men take on the role of market work.

Finally, we quantify the overall social welfare surplus associated with the two scenarios. Figure \ref{Fig_CF_V} plots the per-period welfare under the two scenarios. Under the low-cost scenario, the welfare attained at the end of our prediction period is 17\% higher than under the high-cost scenario.
The welfare increase is aligned with the fact that the high-cost scenario imposes additional constraints on couples who wish to divorce, and deters unmarried couples from formalising their unions. In addition, since removing these frictions is not associated with a decline in specialisation (and the associated positive externalities of marriage), the welfare corresponding to the low-cost scenario is bound to exceed that of the high-cost scenario. 

\begin{figure}[ht!]
  \caption{Counterfactuals: Period Welfare}
\includegraphics[width=0.95\textwidth]{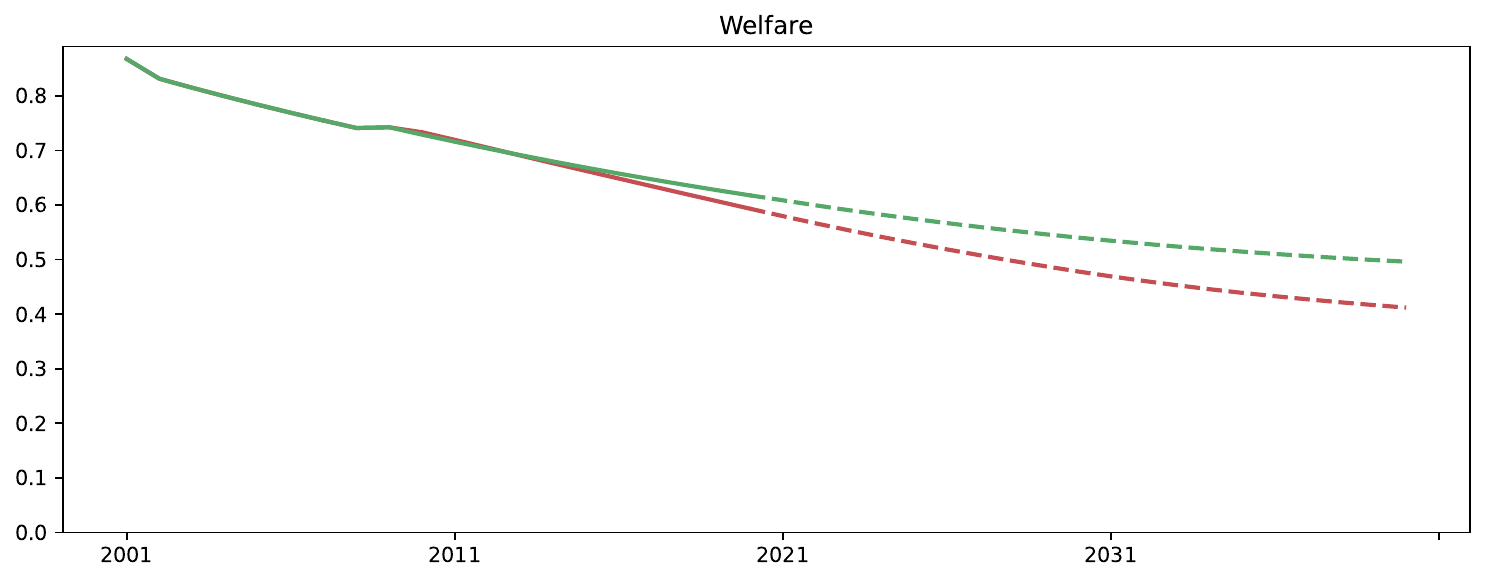}
  \label{Fig_CF_V} 
  \begin{tablenotes}
    The red line corresponds to the 2009 administrative divorce ban scenario, whereas the green line corresponds to the scenario in which administrative divorce remains available. Dashed lines correspond to predictions of marriage market dynamics beyond the span of our data.
  \end{tablenotes}
\end{figure}

\subsection{Discussion}

The parameter estimates and counterfactuals produced by our model help us understand the fine margins through which the administrative divorce option affects the marriage market and its dynamics. The surplus matrices shown in Tables \ref{Table_2x2} and \ref{Table_3x3} indicate that the utility surplus of marriage falls under the low cost scenario ($d=1$). This is to be expected, since the utility surplus of marriage denotes the intrinsic value of sustaining marriage for people who are already married (the value for the insiders). Because of the lower divorce costs, the intrinsic value of sustaining marriage falls, which is reflected in a higher divorce rate under the low-cost scenario.

By contrast, the counterfactual simulations in Figures \ref{Fig_CF_joint_lt} and \ref{Fig_CF_unc} show that the low-cost scenario is also characterised by higher marriage rates. On the surface, this finding may seem inconsistent with the aforementioned decline of the utility surplus of marriage. However, we have to remember that the shares of married couples (along the other couple types) are governed by the transition matrices. This means that the shares are heavily influenced by the behavior of marriage market entrants (and re-entrants). If marriage, as an institution, becomes sufficiently more attractive to the unmarried couples (meaning that its value for the outsiders rises), the model will predict the low-cost scenario to be characterised by both higher marriage rates accross all types and higher divorce rates conditional on being married. 

The specialisation margin of our model uncovers an important source of heterogeneity: the fall of the intrinsic value of marriage for insiders is concentrated among two-earner couples. The male-breadwinner couples sustain the same value of marriage irrespective of the divorce-cost regime, which we attribute to the complexity of their alimony and custody arrangements. With regard to the effects on outsiders, we see greater inflows into both marriage types. The inflow effect is stronger for two-earner couples, which is in line with the fact that they are the main benefactors of the administrative divorce option. The smaller positive effect on male-breadwinner couples likely reflects the fact that they can always transition into the two-earner type and thus become more likely to exercise the administrative divorce option.    

In terms of broader marriage market trends, we show that the administrative divorce option has a substantive dampening effect on the negative trend that characterizes the incidence of marriage in the Netherlands. Nevertheless, the policy alone does not prove to be sufficient to reverse the negative trend. This is perhaps not too surprising: while the implied reduction of divorce costs is certainly meaningful, it is only relevant for the couples whose divorces do not require complex settlement negotiations. Furthermore, irrespective of the divorce regime, many marriages involve non-trivial entry costs associated with wedding celebrations. These costs are shaped largely by cultural factors and are therefore less amenable to public policy interventions. The costs of marriage entry, along with the lack of access to administrative divorce, are the likely reason why the Netherlands is currently experiencing a rise in the popularity of registered partnerships. These types of unions enable couples to formalize their relationships without incurring much of the marriage costs, and therefore constitute a viable alternative pathway for the modern marriage markets.  

\section{Conclusion}
\label{Conclusions}

We exploited an unintended, unanticipated divorce reform to analyse the effects of divorce costs reduction on the marriage rates and broader marriage market dynamics. Using a dynamic structural model of matching decisions, we showed that the marriage rates are higher under the low-cost regime. This finding reflects the fact that the low-cost regime induces more marriage formations than dissolutions. Aggregate welfare also improves, reflecting the removal of financial frictions for people who want to divorce, and retention of household specialisation due to higher marriage rates.

An interesting direction for future research is to extend our model by incorporating having children as a decision variable.
Doing so in standard stationary models is difficult, because fertility is confined to an age-bound window. This makes the childbearing decision inherently age-dependent, and therefore unfit to be represented by a time-invariant policy rule.
Our approach already departs from a steady-state setting, providing a natural foundation to study policy questions with age-specific decisions and outcomes. For example, one could use this framework to study the welfare effects of public investments in childcare systems (either through expansion or subsidisation), taking into account the implied reductions in the costs of children and their effects on fertility, match formation, and marital stability.

\bibliography{submission}

\appendix

\section{Appendix: Institutional Details}\label{sec:institutions}

\subsection{Marriage in the Netherlands}
The commencement of marriage in the Netherlands requires a civil wedding procedure, which is administered by the local municipality. The costs of this procedure depend on a variety of factors, including the day of the week and the type of ceremony. The basic ceremony is very affordable: over the period of observation, its cost rose from approximately \euro 100 to \euro 200, in line with inflation. The costs of more involved wedding ceremonies depend on a variety of factors.\footnote{According to online resources, the wedding celebration costs in the Netherlands averaged \euro 10,000 to \euro 15,000 over the period of observation.} Marriage affords couples a variety of benefits and protections, including tax partnership status, mutual pension rights, a joint property regime, inheritance protection, automatic conferral of parental responsibility, and overall administrative simplicity. During the period of observation, the default matrimonial property regime was full community of property. In 2018, this default was replaced by a limited community of property regime.

The Netherlands uses a no-fault divorce regime, which was introduced in 1971. Divorcing couples are required to seek legal representation, and the divorce procedure is initiated by one of the lawyers filing a petition with the court. There is no mandatory separation period, which means that the petition may be filed irrespective of the couple’s current living arrangements. The costs of divorce are highly idiosyncratic, depending on the amount of time spent with lawyers, the number of court hearings, and other factors. Legal services for low-income families are heavily subsidised, which means that disadvantaged couples can avoid a large share of the financial costs.

\subsection{De facto couples and registered partnerships}
\label{sec:regpar}

Couples who are living out of wedlock (in de-facto cohabitation arrangements) have considerably lower levels of legal protection than married couples. The general rule is that the benefits conferred automatically upon marriage have to be negotiated, contracted, and notarised by de facto couples in order to be legally binding. Compared to countries where de facto couples become legally recognised once they cross a specified relationship-duration threshold (such as Australia, Canada, or New Zealand), the Netherlands does not maintain a similar recognition rule. This means that, regardless of relationship duration, partners in de facto couples do not have an automatic legal right to alimony payments or property division upon relationship breakdown.

In 1998, the Dutch government introduced a new type of civil union: registered partnership. Registered partnerships were targeted at same-sex couples, allowing them to enter a formally recognised civil union while keeping marriage reserved for heterosexual couples only. Different-sex couples were allowed to enter registered partnerships as well, but their interest in partnerships proved relatively low.\footnote{Take-up increased over the period of observation; however, the vast majority of different-sex registered partnerships served as transitory civil unions that were later converted into marriage.} The legal benefits and protections of registered partnership were less comprehensive than those from marriage, but it allowed filing for administrative divorce.

The administrative divorce procedure for registered partners consisted of three steps. First, the couple was required to draft a mutually agreeable divorce settlement. Second, the settlement had to be signed by both parties and notarised. Third, the notarised settlement had to be lodged with the municipality. Administrative divorces were strictly uncontested, which meant that the settlement terms were not challenged by municipal officers, and the divorce became effective immediately upon receipt of the settlement contract. If the registered partners were unwilling or unable to draft a mutually agreeable settlement, they could resort to a conventional divorce procedure resembling the procedure for married couples.

\subsection{Marriage equality reform and its aftermath}

Access to administrative divorce remained exclusive to registered partners for three years. This changed in 2001, when the Dutch government passed a bill legalising same-sex marriage. The bill came into effect on April 1, 2001, and had several consequences for marriages and registered partnerships in the Netherlands. Besides allowing same-sex couples to marry, the bill also legislated the process of transitioning from registered partnership to marriage. Specifically, it allowed registered partners to transform their unions into marriages by means of a simple administrative procedure, under which couples were required to lodge a request with the municipality and pay a small fee. This procedure was quickly exercised by many same-sex couples who preferred to be recognised as married.

However, the bill also allowed reverse transitions. To maintain legal symmetry, married couples were allowed to transition into registered partnerships, and they could do so using the same administrative procedure. This feature of the bill proved to have substantial unintended consequences. Unforeseen by policymakers, reverse transitions quickly became popular, with many (mostly different-sex) couples transforming their marriages into registered partnerships. 

\subsection{Flash divorce}

The reverse transition enabled married couples to switch to the legal regime applicable to registered partners, which included the right to dissolve the partnership by means of administrative divorce. There were no other substantive benefits associated with being in a registered partnership, and nearly all couples who made the reverse transition immediately followed through with their divorce intentions.

Compared to conventional divorce, flash divorce was considerably cheaper, as discussed in the main body of the text, and faster to process. According to legal practitioners, conventional divorces take approximately three months to settle. This reflects the fact that the courts require a minimum of six to eight weeks to process submitted divorce petitions, and the submission itself is preceded by a period during which the spouses meet with their legal teams and prepare for submission. In contrast, flash divorce was not burdened by any statutory processing times, which meant that the couple’s decision could become effective immediately upon notarising the divorce settlement. The administrative burden was therefore greatly reduced, and couples could be divorced in a fraction of the time required under the conventional procedure.

Flash divorce was also more accessible for couples living outside regional centres. Unlike conventional divorces, which are processed by a small network of 19 regional courts, flash divorces were processed by local municipal offices. There are 390 municipalities in the Netherlands. By opting for a flash divorce, couples could save the costs and time associated with commuting to courts and law firms. Given Dutch geography, these considerations were inessential for most couples, but they could prove substantive for couples living in rural areas or on the coastal islands.

\subsection{The ban on the flash divorce option}

Despite its popularity, the institution of flash divorce was contentious. As an unintended outcome of the same-sex marriage bill, flash divorce was poorly incorporated into existing family legislation. It created legal ambiguities, and legislative experts were concerned about its inter-jurisdictional implications (see \citealt{Aygun2015}). Specifically, the flash divorce option was not embedded in international treaties, and couples could face problems when trying to have the flash divorce recognised in their countries of origin. For these reasons, the government passed a corrective bill that prohibited further transitions from marriage into registered partnership, effectively banning the flash divorce option. The bill was passed on November 28, 2008, and came into effect three months later, on March 1, 2009.

\section{Appendix: Modelling Details}
\subsection{General Model}\label{sec:generalmodel}

This appendix section provides the formal microfoundation for the recursive population problem in the main text. 
The main text works with a parsimonious symmetric type space $\mathcal I=\{\mathsf U,\mathsf{CH},\mathsf{CW}\}$ and regime-specific type-level surpluses. 

Here we state the more general two-sided random utility model, derive the entropy-regularised planner representation, and show how the transition operators map current couplings into next period type masses.
To capture heterogeneity in partner choices we allow the utility a person derives from being matched to their partner to depend on the specific match.
Contemporaneous utilities have a common component depending on observable characteristics (which we seek to infer from the distribution of matches), and an unobserved component drawn from a known distribution.

Further, for each side of the market's population there exists an equivalence relation which induces finite partitions $ \mathcal{M} $ and $ \mathcal{F} $, respectively, such that any individual on one side of the market is indifferent over any two individuals within the same equivalence class of the opposite side of the market.
This allows discrete types and makes matching tractable.
Assumption \ref{ass:TU} formalises the standard assumptions. %

\begin{ass}
\label{ass:TU}
For any equivalence class (which we refer to as ``individual type'') $ m \in \mathcal{M} $, $ f \in \mathcal{F} $, and time period $ t \in \mathcal{T} $ let $ it \in m $ and $ jt \in f $ be two individuals from opposite sides of the marriage market. Let $ \tilde{U}_{ijt} $ and $ \tilde{V}_{ijt} $ be their structural latent utilities and $ z_{it} $ and $ z_{jt} $ their observed characteristics.
    \begin{enumerate}
      \item Additive random utility: $ \tilde{U}_{ijt} =  U^0(z_{it}, z_{jt}) + \varepsilon_{ijt} $ and $  \tilde{V}_{ijt} = V^0(z_{it}, z_{jt}) + \eta_{ijt} $,
      \item Individuals have preferences over their partner's type only: 
        \begin{enumerate}
          \item Heterogeneity is  $  (z_{it}, z_{jt}) = (z_{mt}, z_{ft}) $ for all $ it \in m, jt \in f $,
          \item $ (\varepsilon_{ijt}, \eta_{ijt})  = (\varepsilon_{ift}, \eta_{mjt}) $ with random vectors $ \bm{\varepsilon}_{it} = (\varepsilon_{ift})_{f \in \mathcal{F}} $ and $ \bm{\eta}_{jt} = (\eta_{mjt})_{m \in \mathcal{M}} $,\footnote{This restricts the dimension of the unobserved heterogeneity for all individuals of the same type: $ \varepsilon_{ij} = \varepsilon_{ij'} $ for all $ j,j' \in f, $ and $ \eta_{ij} = \eta_{i'j} $ for all $ i,i' \in m $. We write $ \varepsilon_{if} $ and $ \eta_{mj} $ shorthand for the respective $ \varepsilon_{ij} $ and $ \eta_{ij}  $.}
      \end{enumerate}
    \item The distributions of $ \bm{\varepsilon}_{it} $  and $ \bm{\eta}_{it} $ are heteroskedastic and conditionally time-independent: $ P_m(\bm\varepsilon_{it} | m_{t}, m_{t-1}, \bm\varepsilon_{it-1}) = P_m(\bm\varepsilon_{it} | m_{t}) $ and $ Q_f(\bm\eta_{jt} | f_{t}, f_{t-1}, \bm\eta_{jt-1}) = Q_f(\bm\eta_{jt} | f_{t}) $,
      \item Transferable Utility: $ \tilde{\Phi}_{ijt} = \tilde{U}_{ijt} + \tilde{V}_{ijt} \equiv \Phi_{mft} + \varepsilon_{ift} + \eta_{mjt} $ where $ \Phi_{mft} = U^0_{mft} + V^0_{mft} = U^0(z_{mt}, z_{ft}) + V^0(z_{mt}, z_{ft}) $,
        \item Normalisation of the ``remaining single'' outside options: $ (\tilde{\Phi}_{i0t}, \tilde{\Phi}_{0jt}) = (\varepsilon_{i0t}, \eta_{0jt}) $.
    \end{enumerate}
\end{ass}

The resulting estimation problem is an inverse problem, as is often the case in structural modelling.
We observe the outcome, in this case the equilibrium distribution of matches, and want to learn which systematic utilities $ U_{mft} $ and $ V_{mft} $ rationalise this outcome.
Formally, the matching outcome is characterised by a matrix  $ \bm{\bar{\mu}}_t $ which collects the equilibrium market shares $ \mu_{mft}$ of couples of type $ mf \in \mathcal{M} \times \mathcal{F} $ at time $ t $.
Furthermore, we collect unmatched men into the vector $ \bm\mu_{\bm{\cdot}0t} = (\mu_{m0t})_{m\in \mathcal{M}} $ and unmatched women into the vector into $ \bm\mu_{0\bm{\cdot}t} = (\mu_{0ft})_{f \in \mathcal{F}} $,  and define:

\begin{dfn}\label{dfn:couplings}
  The couplings matrix is the equilibrium, observed outcome of the matching market in period $ t \in \mathcal{T} $.
\begin{equation}
\bm{\mu}_t = \left[\begin{matrix} \bm{\bar{\mu}}_{t} & \bm\mu_{\bm{\cdot}0t} \\ \bm\mu_{0\bm{\cdot}t}^T & 0 \end{matrix} \right].
\end{equation}
\end{dfn}

Individuals maximise their structural utility net of a match-specific transfer they pay to (or receive from) their partner in equilibrium: $ U_{mf}(\bm\mu_t) = U^0_{mf} + \tau_{mf}(\bm\mu_t) $ and $ V_{mf}(\bm\mu_t) = V^0_{mf} - \tau_{mf}(\bm\mu_t)$.
Their indirect utilities are thus given by:
$$ U_{mt}^*(\bm\mu_t) = \mathbb{E}_{P_m} \max\limits_{f \in \mathcal{F}} \{ U_{mf}(\bm\mu_t) + \bm\varepsilon_{ft} \} \qquad  V_{ft}^*(\bm\mu_t) = \mathbb{E}_{Q_f} \max\limits_{m \in \mathcal{M}} \{ V_{mf}(\bm\mu_t) + \bm\eta_{mt} \}.  $$

Let $ \bm{M}_t $, and $ \bm{F}_t $  be probability measures supported on $ \mathcal{M} $ and $ \mathcal{F} $ with $ M_{mt} $, and $ F_{ft} $ measuring the events $ \{ it \in m \} $ and $ \{jt \in f \} $, respectively.
Any matching allocation $ \bm\mu_t $ must satisfy the following feasibility constraint
\begin{equation}
  \Gamma(\bm{M}_t, \bm{F}_t) = \left\{ \bm\mu \in \Delta \; \Bigg| \; \forall m \in \mathcal{M} :  \sum\limits_{f \in \mathcal{F} \cup \{0\} } \mu_{mft} = M_{mt} \text{ and } \forall f \in \mathcal{F} :  \sum\limits_{m \in \mathcal{M} \cup \{0\} } \mu_{mft} = F_{ft} \right\}
\end{equation}
where $ \Delta $ is the $ ((|\mathcal{M}|+1)(|\mathcal{F}|+1)-1)$-dimensional unit-simplex. We can immediately write down the expected utilitarian matching surplus of the population.
\begin{align}
\text{Welfare}_t
    &=\max\limits_{\bm\mu_t \in \Gamma(\bm{M}_t,\bm{F}_t)} \left\{ \sum\limits_{m \in \mathcal{M}} M_{mt} U_{mt}^*(\bm\mu_t) +
    \sum\limits_{f \in \mathcal{F}} F_{ft} V_{ft}^*(\bm\mu_t) \right\} \\
  \label{eq:primal} &= \max\limits_{\mu_t \in \Gamma(\bm{M}_t, \bm{F}_t)} \left\{ \sum_{mf \in \mathcal{M}\times\mathcal{F}} \mu_{mft} \Phi_{mft} + \mathcal{E}(\bm\mu_t, \bm{M}_t, \bm{F}_t) \right\} 
\end{align}
where the second equality is from \citet{Galichon2012}, who show equivalence to an entropy-regularised optimal transport problem and prove strong duality.\footnote{The generalised entropy term $ \mathcal{E}(\mu, M, F) = -\sum_{m \in \mathcal{M}} M_{mt} \mathbb{E}_{P_m} \varepsilon_{ift} - \sum_{f \in \mathcal{F}} F_{ft} \mathbb{E}_{Q_f} \eta_{mjt} $ arises due to the additive separability of the ARUM. With logistic errors the term becomes $\mathcal{E}_{\sigma}(\mu, M, F) =   -\sum_{m} M_{mt} \sigma_{m} H(P_m) - \sum_{f} F_{ft} \sigma_{f} H(Q_f) $, where $ H(P) = -\sum_{\ell} P_{\ell} \log(P_{\ell})$ is the Shannon entropy. If all errors are Gumbel $ P = P_m = Q_f $ it has the familiar form from \citet{Choo2006}. }

We let $ \varepsilon $ and $ \eta $ be Gumbel with type-specific scales $ \sigma_m $ and $ \sigma_f $, respectively.
Then the entropy term reduces to a combination of Shannon entropies weighted by the scales of the unobserved heterogeneity.
Intuitively, $ \mathcal{E} $ rationalises the unobserved heterogeneity in the data by regularising the joint distribution $ \mu_t $ towards a uniform distribution, i.e. random matching.
We allow for different scales because we see very different matching patterns between married and unmarried types.\footnote{Being single and unmarried is more attractive than being divorced, which is difficult to rationalise using the same scale of unobserved heterogeneity.}
The deterministic part of the overall welfare can be written as $  \sum_{mf \in \mathcal{M}\times\mathcal{F}} \mu_{mft} \Phi_{mft} = \text{tr}(\bm{\Phi}^T \bm{\bar{\mu}}_t) $.
We parameterise the match surplus $ \bm{\Phi} $ by a sieve with the appropriate basis functions for square matrices (symmetric, assortative, fully saturated) to get
\begin{equation}
    {\Phi}_{\bm{\gamma}} \in \left\{ \sum_{k = 1}^K \gamma_k \phi_k  : \phi \in \mathfrak{B} \right\}.
\end{equation}
where we stack $ \gamma $ and also the scales of $ \mathcal{E}_{\sigma} $ into the parameter vector ${\theta} $ and write $ \bm{\Phi}_{\theta} $ and $ \mathcal{E}_{\theta} $ to highlight this dependence.

To incorporate dynamics we now let the types evolve over time and consider the planner's intertemporal problem.
From Assumption \ref{ass:TU} we have $ \bm\varepsilon_{it} $ and $ \bm\eta_{jt} $ that are mutually independent within and across time.
Thus, in each period a new copy of unobserved preferences is realised, making the previous allocation of $\bm{\mu}$ suboptimal.

With this law of motion, we can define the process of state-action pairs $ \bm{\tau}_t = ((\bm{M}_t, \bm{F}_t), \bm{\mu}_t)$ under policy $ \bm\mu_t \sim \pi_{\theta, \gamma}(\bm{\cdot} | \bm{M}_t, \bm{F}_t) \in \Gamma(\bm{M}_t, \bm{F}_t) $. 
\begin{equation}
  \left\{\bm{\tau}\right\}_{t = 0}^{\Infty} = (
  \overbrace{((\bm{M}_0, \bm{F}_0), \bm\mu_0)}^{\bm\tau_0}, 
  \overbrace{(\underbrace{\Lambda(\bm{\mu}_0, \xi_0)}_{(\bm{M}_1, \bm{F}_1)},
  \underbrace{\pi_{\bm{\theta}}(\bm{M}_1, \bm{F}_1)}_{\bm{\mu}_1})}^{\bm{\tau}_1}, 
  \overbrace{(\underbrace{\Lambda(\bm{\mu}_1, \xi_1)}_{(\bm{M}_2, \bm{F}_2)}, 
  \underbrace{\pi_{\bm{\theta}}(\bm{M}_2, \bm{F}_2)}_{\bm{\mu}_2})}^{\bm{\tau}_2},
  \ldots).
\end{equation}
Since our policy is determined by the structure of the matching problem (which is deterministic), the policy $ \pi_{\theta,\gamma} $ is degenerate.
$ \gamma $ is a nuisance parameter arising from the approximation of the policy rule solving the dynamic choice problem.
We defer a more detailed discussion of its role until we have set up the dynamic problem.
The likelihood of a sequence of state-action pairs $\bm{\tau}_t$, given $ \theta $ and $ \gamma $ is
\begin{equation}
  \label{eq:likelihood}
  \nu(\bm{\tau}; \theta, \gamma) = \prod_{t=0}^{\Infty} \pi_{\theta, \gamma}(\bm{\mu}_t|\bm{M}_t,\bm{F}_t) P_{\xi}(\xi_t)  \delta((\bm{M}_{t+1},\bm{F}_{t+1}) - \Lambda(\bm{\mu}_t,\xi_t)).
\end{equation}

The planner's objective is to re-optimise to a new, inter-temporally efficient outcome by finding a matching-allocation  that is feasible for the current and expected future masses of types.
The discounted lifetime value of welfare $R$ corresponding to the given state-action process $ \bm\tau $ is:
\begin{equation}
  R(\bm{\tau};\theta) = \sum\limits_{t=0}^{\Infty} \beta^t r(\bm{\tau}_t;\theta) 
  = \sum\limits_{t=0}^{\Infty} \beta^t \left( \text{tr}(\bm{\Phi}_{\theta}^T \bm{\bm{\bar{\mu}}}_t) + \mathcal{E}_{\theta}(\bm{\bm{\mu}}_t, \bm{\bm{M}}_t, \bm{F}_t) \right)
\end{equation}

With the expected lifetime welfare under a given policy $ {\pi}_{\theta,\gamma} $ we can define the state-action value function:
\begin{equation}
  \label{eq:stateactionvalue}
  Q(\bm{\tau};\theta, \gamma) = \int_{\bm{\tau}} R(\bm{\tau};\theta) \nu(\bm{\tau};\theta,\gamma) = r(\bm{\tau}_0;\theta) + \beta \int_{\bm{\tau}} \sum\limits_{t=1}^{\Infty} \beta^{t-1} r(\bm{\tau}_t;\theta) \nu(\bm{\tau};\theta,\gamma)
\end{equation}

Conditioning on the initial state $ \tau_0 $ we get (slightly misusing the notation):
\begin{align}
  Q(\bm{\tau};\theta, \gamma) &= r(\bm{\tau}_0;\theta) + \beta \int_{\bm{\tau}_1 = L'\bm{\tau}} R(\bm{\tau}_1;\theta) \nu(\bm{\tau}_1|\bm{\tau}_0; \theta, \gamma)  \\
  &= r(\bm{\tau}_0;\theta) + \beta \int_{\bm{\tau}_1}\left[ r(\bm{\tau}_1) + \beta \int_{\bm{\tau}_2=(L')^2 \bm{\tau}} R(\bm{\tau}_2;\theta)\nu(\bm{\tau}_2|\bm{\tau}_1, \bm{\tau}_0; \theta, \gamma)\right] \nu(\bm{\tau}_1|\bm{\tau}_0; \theta, \gamma)  \\
  &= r(\bm{\tau}_0;\theta) + \beta\int_{\bm{\tau}_1}\left[  r(\bm{\tau}_1) + \beta \int_{\bm{\tau}_2=(L')^2 \bm{\tau}} R(\bm{\tau}_2;\theta)\nu(\bm{\tau}_2|\bm{\tau}_1;\theta, \gamma)\right] \nu(\bm{\tau}_1|\bm{\tau}_0; \theta, \gamma) \\
&= r(\bm{\tau}_0;\theta) + \beta\int_{\bm{\tau}_1}  Q(L'\bm{\tau}; \theta, \gamma) \nu(\bm{\tau}_1|\bm{\tau}_0; \theta, \gamma),
\end{align}
where $ L' $ is the forward-lead operator.
The penultimate equation follows from conditional time-independence in Assumption \ref{ass:TU}.
Making the state and action-transitions explicit using our policy function and structure of collective decisions, we can write:
\begin{align*}
    Q(\bm{M}, \bm{F},&\bm{\mu};\theta,\gamma) = r(\bm{M}, \bm{F},\bm{\mu}) + \beta\int_\xi \left[\int_{\bm{\mu}'} Q(\Lambda(\bm{\mu},\xi),\bm{\mu}';\theta,\gamma)\pi_{\theta,\gamma}(\bm{\mu}'|\Lambda(\bm{\mu}, \xi))\right]dP_{\xi} \\
                     &= r(\bm{M}, \bm{F},\bm{\mu}) + \beta\int_{(\bm{M},\bm{F})'}\left[\int_\xi \delta((\bm{M},\bm{F})'-\Lambda(\bm{\mu},\xi))dP_{\xi} \right]\\
   &\hspace{4em} \times \left[\int_{\bm{\mu}'} Q((\bm{M},\bm{F})',\bm{\mu}';\theta, \gamma)\pi_{\theta,\gamma}(\bm{\mu}'|(\bm{M},\bm{F})') \right]d(\bm{M},\bm{F})' \\
   &= r(\bm{M}, \bm{F},\bm{\mu}) + \beta\int_{(\bm{M},\bm{F})'} 
   Q((\bm{M}',\bm{F}'),\bm\mu^*_{\theta,\gamma}(\bm{M}',\bm{F}');\theta,\gamma) d\mathcal{P}^*((\bm{M},\bm{F})'|\bm{M},\bm{F},\bm{\mu}) 
\end{align*}
where the last equality follows from the structure of collective decisions and the fact that we have a deterministic policy:
\begin{equation}
    \int_{\bm{\mu}'}
      Q((\bm{M}',\bm{F}'),\bm{\mu}';\theta,\gamma)\,
      \pi_{\theta,\gamma}(\bm{\mu}' \mid (\bm{M}',\bm{F}')) = 
    Q((\bm{M}',\bm{F}'),\bm\mu^*_{\theta,\gamma}(\bm{M}',\bm{F}');\theta,\gamma).
\end{equation}

For identification (discussed in the next section) we use the result from \citet*{Corblet2023}, who show that the competitive equilibrium of this repeated infinite-horizon matching model is equivalent to the planner's problem above.
They marginalise the joint transition probability $ \mathcal{P}^* : (\bm{M}, \bm{F}) \mapsto (\bm{M}', \bm{F}') $ to $ \mathbf{P}^m : (\bm{M}, \bm{F}) \mapsto \bm{M}' $ and $ \mathbf{P}^f :(\bm{M}, \bm{F}) \mapsto \bm{F}' $, respectively, allowing the decision of one partner to depend only on the state of the couple and not on the (simultaneous) decision of the partner.
We want to invoke their results and follow this approach.
We interpret this restriction as a compromise between completely unrestricted evolution of both types and independent evolution of men's and women's types, respectively.\footnote{Evaluating transition probabilities can be viewed as a tensor-contraction. Using Einstein sum notation we obtain $ \mathbf{P}^{m'\: m}_{\quad \quad f} (\bm{J}_m\bm{\mu})^{f}_{\; \; m} = \bm{M}^{m'} $, where $ \mathbf{P} \in \mathbb{R}^{|\mathcal{M}|+1} \otimes \mathbb{R}^{|\mathcal{M}|+1} \otimes  (\mathbb{R}^{|\mathcal{F}|} \rightarrow \mathbb{R})  $ is defined as a (1,2)-tensor, $ \bm{\mu} \in \mathbb{R}^{(|\mathcal{M}|+1) \times |\mathcal{F}|} $ is a (1,1)-tensor (matrix), and $ \bm{M} \in \mathbb{R}^{|\mathcal{M}|+1} $ a (0,1)-tensor (column-vector) where $ (\bm{\cdot}, \bm{\cdot}) $ denote contra- and covariant dimensions, respectively. Intuitively, the covariant dimension of the tensor is responsible for the marginalisation with respect to $ \mathcal{F} $. The tensor-based definition proves very useful for the \texttt{fold} operation in the computational part, for which we leverage the tensor-library \texttt{PyTorch}. For the remainder of this section, we vectorise $ \Vec ( \bm{J}_m \bm{\mu}) \in \mathbb{R}^{(|\mathcal{M}|+1)|\mathcal{F}|} $ which leads to the operation $ \mathbf{P}^{m'}_{\quad mf} (\Vec (\bm{J}_m \bm{\mu}))^{mf} = \bm{M}^{m'} $ equivalent to equations \eqref{eq:transPQ}. We can proceed similarly for $ \bm{Q} $.} 

\begin{ass}[Marginalisation of Transition Probabilities]
  \label{ass:transprob}
  Transitions from match types to individual types follow the following linear Markovian structure:\footnote{It is easy to show that $ \mathcal{P}^m : \bm\mu_t \mapsto (\mathbf{P}^m \otimes \bm{J}_m) \bm{K} \Vec \bm{\mu}_t $ and  $ \mathcal{P}^f : \bm\mu_t \mapsto (\bm{J}_f^T \otimes \mathbf{P}^f ) \Vec \bm{\mu}_t $ where $ \bm{K} $ is the commutation matrix and $ \bm{J}_m = [\bm{I}_{|\mathcal{M}| -1}, \mathbf{0} ]$, $ \bm{J}_f = [\bm{I}_{|\mathcal{F}| - 1}, \mathbf{0}]^T $.}
\begin{equation}
    \mathcal{P}^m \bm{\mu}_t = \bm{M}_{t+1} \qquad 
    \mathcal{P}^f \bm{\mu}_t = \bm{F}_{t+1}. \label{eq:transPQ}
\end{equation}
\end{ass}
This implies 
$  \mathcal{P}^*((\bm{M}',\bm{F}') \mid \bm{M},\bm{F},\bm{\mu}) = \delta_{(\mathcal{P}^m\bm{\mu},\mathcal{P}^f\bm{\mu})}(\bm{M}',\bm{F}')$, 
and thus:
\begin{align}
  \int_{(\bm{M}',\bm{F}')}
    \,
    Q\big((\bm{M}',\bm{F}'),\bm\mu^*_{\theta,\gamma}(\bm{M}',\bm{F}');\theta,\gamma\big)\,
\delta_{(\mathcal{P}^m\bm{\mu},\mathcal{P}^f\bm{\mu})}(\bm{M}',\bm{F}') \\
=
  Q\big((\mathcal{P}^m\bm{\mu},\mathcal{P}^f\bm{\mu}), \bm\mu^*_{\theta,\gamma}(\mathcal{P}^m\bm{\mu},\mathcal{P}^f\bm{\mu});\theta,\gamma\big).
\end{align}

Due to the feasibility constraint $ \bm\mu^*_{\theta,\gamma}(\bm{M}, \bm{F}) \in \Gamma(\bm{M}, \bm{F}) $ this significantly simplifies the Bellman equation to
\begin{equation}
  \label{eq:bellman}
  Q(\bm{M}, \bm{F}, \bm{\mu}; \theta, \gamma) =
  \underbrace{\text{tr}(\bm{\Phi}_{\theta}^T \bm{\mu}) +
  \mathcal{E}_{\theta}(\bm{\mu}, \bm{M}, \bm{F})}_{r_\theta(\bm M, \bm F, \bm \mu)} +
        \beta Q(\mathcal{P}^m\bm{\mu}, \mathcal{P}^f\bm{\mu}, \bm{\mu}; \theta, \gamma).
\end{equation}

Using the envelope theorem, taking first order conditions leads to the following equilibrium restrictions defining our model:
\begin{align}
    \label{eq:bellmanresid}
    \mathcal{R}(\bm{M}, \bm{F}, \bm{\mu}; \theta, \gamma) &\equiv Q(\bm{M}, \bm{F}, \bm{\mu}; \theta, \gamma) - r_\theta(\bm{M}, \bm{F}, \bm{\mu}) - \beta Q(\mathcal{P}^m\bm{\mu}, \mathcal{P}^f\bm{\mu}, \bm{\mu}; \theta, \gamma) = 0, \\
    \label{eq:foc}
    \mathcal{F}(\bm{M}, \bm{F}, \bm{\mu}; \theta, \gamma) &\equiv \frac{dQ(\bm{\mu})}{d \Vec \bm{\mu}^T} = \frac{d\text{tr}(\bm{\Phi}_{\theta}^T \bm{\mu})}{d \Vec \bm{\mu} ^T} + \beta \left( \frac{dQ}{d \bm{M}^T} \frac{d\bm{M}}{d \Vec \bm{\mu}^T } + 
    \frac{dQ}{d \bm{F}^T} \frac{d\bm{F}}{d \Vec \bm{\mu}^T }\right) = \mathbf{0},
\end{align}

We define the equilibrium of the matching market directly.
The allocation $\bm{\mu}_t$ and transfers $(\tau_{mf,t})$ form a dynamic
competitive equilibrium if, given the transfers, every individual's
assignment under $\bm{\mu}_t$ maximises their expected lifetime utility,
inclusive of continuation values, over the feasible alternatives in Table
\ref{tab:feasible} and the outside option, and $\bm{\mu}_t$ is consistent
with the population masses $(\bm{M}_t, \bm{F}_t)$.
Our identifying restriction is that the observed allocation constitutes
such an equilibrium in every period.
In the static model, this requirement is equivalent to stability: no pair
could dissolve their matches and re-match with a transfer that leaves both
strictly better off \citep{Shapley1971, Becker1973}.
Theorem 4 of \citet*{Corblet2023} extends the equivalence with the
planner's problem to the repeated game, which we invoke.

The intuition is that the marginalisation of the transition kernel in
Assumption \ref{ass:transprob} splits the dynamics into a male and a female
part, so that the male and female Bellman equations aggregate additively
into the overall Bellman equation.
Due to transferability of utility, linearity of both the trace and the entropy (with respect to $ \bm{M} $ and $ \bm{F} $), we have $ r_\theta(\bm{M}, \bm{F}, \bm{\mu}) = r^m_\theta(\bm{M}, \bm{F}, \bm{\mu}) + r_\theta^f(\bm{M}, \bm{F}, \bm{\mu}) $.
In conjunction with  the separability from the marginalisation $ \mathcal{P}((\bm{M},\bm{F})'| (\bm{M}, \bm{F}), \bm{\mu})  = \mathcal{P}^m(\bm{M}'|\bm{M}, \bm{F}) \mathcal{P}^f(\bm{F}'|\bm{M}, \bm{F}) %
$ the Bellman equation in \eqref{eq:bellman} can be rewritten as:
\begin{align}
  Q(\bm{M}, \bm{F},\bm{\mu};\theta) &= r^m_\theta(\bm{M}, \bm{F},\bm{\mu}) + \beta \mathcal{P}^m(\bm{M}'|\bm{M},\bm{F},\bm{\mu}) Q^m(\bm{M}',\bm{\mu}';\theta) \\
 &+ r^f_\theta(\bm{M}, \bm{F},\bm{\mu}) + \beta \mathcal{P}^f(\bm{F}'|\bm{M},\bm{F},\bm{\mu}) Q^f(\bm{F}',\bm{\mu}';\theta).
\end{align}

\subsection{Computation}
\label{app:computation}
With many continuous states and actions, solving the recursive problem is difficult.
Nested fixed-point methods are inefficient and prone to numerical instability, whereas standard projection-based methods are limited in their ability to approximate high-dimensional continuous states spaces. 
Further, due to the recursive nature of the problem, while $ \frac{d\bm{F}}{d \Vec \bm{\mu}^T } $ and $ \frac{d\bm{M}}{d \Vec \bm{\mu}^T } $ are easy to obtain from Assumption \ref{ass:transprob}, the derivatives $ \frac{dQ}{d\bm{M}^T} $ and $ \frac{dQ}{d\bm{F}^T} $ are not straight-forward to calculate.
Thus, we will make use of a neural net to approximate the problem and calculate these gradients.

For this, we set up a multi-layer neural-net that approximates both the value function and the policy rule simultaneously: $ \Xi_{\theta,\gamma} : (\bm{M}, \bm{F}) \mapsto (\bm\mu^*_{\theta,\gamma}, Q_{\theta,\gamma}) $.
It takes simulated states $ (M, F) $ and outputs the policy function $ \bm\mu^*_{\bm{\gamma},\bm{\theta}} $ and the value function $ Q_{\bm{\gamma},\bm{\theta}} $, both now implicitly parameterised by the structural parameters $ {\theta} $ and some high-dimensional nuisance parameter $ {\gamma} $ from the approximation.

We define the loss as a weighted $\ell_2$-residual of the derivative of the Lagrangian  \eqref{eq:foc} and the Bellman residual defined in \eqref{eq:bellmanresid} \citep*{maliar2021}.
\begin{equation}
  \widehat{{\gamma}}(\theta) = \arg\min_{\bm{\gamma}} w_1 \left\Vert \mathcal{R}(\bm{M}, \bm{F}, \bm{\mu}; \theta, \gamma) \right\Vert^2  + w_2 \left\Vert \mathcal{F}(\bm{M}, \bm{F}, \bm{\mu}; \theta, \gamma) \right\Vert^2.
\end{equation}

\begin{algorithm}
  \caption{Projection-Based Algorithm for Bellman Equation in \eqref{eq:bellmanresid}}
  \setstretch{1.25}
\begin{algorithmic}[1]
\REQUIRE Initialise parameters $\kappa$, $\gamma$, $\beta$, $w_1$, $w_2$, $N$, $\epsilon$, $t \gets 0$
\ENSURE Optimised parameters and convergence
\REPEAT
\STATE Draw $\bm{M}, \bm{F} \sim \text{Dirichlet}(\kappa \cdot \boldsymbol{\iota}_{K})^{N}$ with $ K = |\mathcal{M}|, |\mathcal{F}| $
\STATE Matching $(\bm{\mu}, Q_{\theta, \gamma}) = \Xi_{\theta, \gamma}(\bm{M}, \bm{F}) $
    \STATE Household choices $(\bm{M}', \bm{F}') = (\mathcal{P}^m \bm{\mu}, \mathcal{P}^f \bm{\mu})$
    \STATE Matching $(\bm\mu', Q_{\theta,\gamma}') = \Xi_{\theta,\gamma}(\bm{M}', \bm{F}') $
    \STATE Prepare for \texttt{autograd}: $\tilde{Q}_{\theta,\gamma} = \sum_{g = 1}^{N} (r_{\theta,\gamma,g} + \beta \cdot Q_{\theta,\gamma,g})$
    \STATE Calculate loss: $l(\bm{\gamma}) =  w_1 \cdot \left\|Q_{\theta,\gamma} - r_{\theta,\gamma} - \beta Q_{\theta,\gamma}'\right\|^2 + w_2 \cdot \left\|\frac{\partial \tilde{Q}_{\theta,\gamma} }{\partial \Vec \bm{\mu}}\right\|^2$
    \STATE Calculate $ \nabla_{\bm{\gamma}} \ell(\bm{\gamma}) $ by backpropagation and take step.
    \STATE $t \gets t + 1$
\UNTIL{$|\ell(\gamma)_t - \ell(\gamma)_{t-1}| < \epsilon$ or maximum iterations reached}
\RETURN optimised parameters $ \gamma $ describing model $ \bm\mu^{*} $ 
\end{algorithmic}
\end{algorithm}

Two remarks are in order. 
First, to obtain gradients that enter the value function directly, rather than the backpropagated gradients, we use automatic differentiation which is readily available in many tensor libraries such as \texttt{PyTorch} or \texttt{Tensorflow}.

Second, we need to ensure that couplings add up to the marginals. 
The entropic regularisation does not in general ensure this.
Enforcing this with hard constraints during optimisation is difficult for neural networks because the relationship between the parameters $ \gamma $ and the outcome is opaque.
To circumvent this, for the final layer of $ \Xi $ we use the Sinkhorn matrix scaling algorithm \citep{sinkhorn1964} to ensure that the marginals are satisfied:\footnote{This can be compared to a two-dimensional \texttt{softmax} function. It is a continuous and smooth operation and thus differentiable.} 
$$ \Xi_{\theta,\gamma}(\bm{M}_{t}, \bm{F}_{t}) \in \Gamma(\bm{M}_{t}, \bm{F}_{t}). $$

\begin{algorithm}
\caption{Sinkhorn Matrix Scaling Algorithm}
\begin{algorithmic}[1]
\REQUIRE Output: $\bm{Y} \in \mathbb{R}^{|\mathcal{M}| \times |\mathcal{F}|}$, $\bm{M} \in \mathbb{R}^{|\mathcal{M}|}_{+}$, and $\bm{F} \in \mathbb{R}^{|\mathcal{F}|}_{+}$
\ENSURE Scaled matrix $\bm{\mu}$ such that $\bm{\mu}\bm{\iota}_{|\mathcal{F}|} = \bm{M}$ and $\bm{\mu}^T\bm{\iota}_{|\mathcal{M}|} = \bm{F}$
\STATE $\bm{\mu}^{(0)} \gets \exp{\bm{Y}}$ \COMMENT{component-wise exponentiation}
\STATE $t \gets 0$
\REPEAT
    \STATE $\bm{d}^{(t)} \gets \bm{M} \oslash (\bm{\mu}^{(t)}\bm{\iota}_{|\mathcal{F}|})$ \COMMENT{element-wise division}
    \STATE $\bm{\mu}^{(t+\frac{1}{2})} \gets \text{diag}(\bm{d}^{(t)})\bm{\mu}^{(t)}$ \COMMENT{apply row scaling}
    \STATE $\bm{e}^{(t)} \gets \bm{F} \oslash (\bm{\mu}^{(t+\frac{1}{2})})^T\bm{\iota}_{|\mathcal{M}|}$ 
    \STATE $\bm{\mu}^{(t+1)} \gets \bm{\mu}^{(t+\frac{1}{2})}\text{diag}(\bm{e}^{(t)})$ \COMMENT{apply column scaling}
    \STATE $t \gets t + 1$
    \STATE $\text{err} \gets \max\{\|\bm{\mu}^{(t)}\bm{\iota}_{|\mathcal{F}|} - \bm{M}\|_\Infty, \|(\bm{\mu}^{(t)})^T\bm{\iota}_{|\mathcal{M}|} - \bm{F}\|_\Infty\}$
\UNTIL{$\text{err} < \epsilon$}
\RETURN $\bm{\mu}^{(t)}$
\end{algorithmic}
\end{algorithm}

This algorithm is very efficient and convergence is often obtained after few iterations.
Because we have to retain gradients through this iterative scaling, we truncate the number of iterations in training and only use a convergence criterion in evaluation.

\subsection{Maximum Likelihood Estimation of Structural Parameters}
\label{app:estimation}
So far, we have optimised for a fixed parameter $ \bm{\theta} $, and thus a given surplus $ \bm{\Phi}_{{\theta}} $.
We are, however, interested in the inverse problem: we observe an outcome $ \bm{\mu} $ and aim to estimate $ \bm{\theta} $.
We can do this by matching the distribution of observed actions or states to the optimal distribution according to our model, varying the parameter $ \bm{\theta} $.
Because state and actions are finite-dimensional (and continuous), matching the probability of each type of individual (states) or each type of match (action) to the respective empirical probability mass function this boils down to moment-matching\footnote{See \citet{Galichon2012,GalichonSalanie2022} for an extended discussion for the static case.}

What makes this problem non-trivial, is that the objective is not jointly convex.
To see this, note that the outer loss depends on $ \theta $ through $ \gamma $ (the optimiser of the inner program) and thus: $ \ell(\gamma, \theta) =  m(\gamma(\theta)) $.
Letting $G({\gamma}, {\theta}) = \frac{\partial \ell({\gamma},\theta)}{\partial \gamma} = 0$, we note, by the implicit function theorem, that $G(\gamma,\theta) = 0 \iff \frac{d\gamma}{d\theta} =  -\left(\frac{\partial G}{\partial \gamma}\right)^{-1} \frac{\partial G}{\partial \theta} $. If $ \frac{\partial G}{\partial \gamma} $ exists, we get $\frac{d\gamma}{d\theta} = -\left(\frac{\partial^2 \ell}{\partial \gamma^2}\right)^{-1} \frac{\partial^2 \ell}{\partial \gamma \partial \theta}$
through which we can obtain the gradient of our outer procedure:
$$\frac{dm}{d\theta} = \underbrace{\frac{\partial m}{\partial \gamma}}_{\text{backpropagation}} \cdot \frac{d\gamma}{d\theta} = - \frac{\partial m}{\partial \gamma} \left(\frac{\partial^2 \ell}{\partial \gamma^2}\right)^{-1} \frac{\partial^2 \ell}{\partial \gamma \partial \theta}.$$
Inverting the Hessian with respect to the nuisance parameter $ \gamma $ is not feasible because of its high dimensionality.
However, we can get $ dm/d\theta $ if we retain the gradient $ d\gamma/d\theta $ from the inner program, at least along the path of descent of the outer program allowing us to solve the program by zig-zagging between the dimensions of $ \theta $ and $ \gamma $.

Let $ (T, |\mathcal{M}|+1, |\mathcal{F}|+1) $ be the dimensions of $\hat{\boldsymbol{\mu}}$.
Recall that $ \Xi_{\theta,\gamma} = (\bm\mu_{\theta,\gamma}, Q_{\theta,\gamma}) $, both functions.
We can, recursively, define an optimal policy process for a fixed $ \theta \in \Theta $:
\begin{align*}
    \bm\mu_{\theta,\gamma,0}^* &= \widehat{\bm{\mu}}_0  \\
    \bm\mu_{\theta,\gamma,t+1}^* &= \bm\mu_{\theta,\gamma}(\bm{M}^*_{\theta,\gamma,t+1}, \bm{F}^*_{\theta,\gamma,t+1})
\end{align*}
where $ \bm{M}^{*}_{\theta,\gamma,t+1} = \mathcal{P}^m \bm\mu_{\theta,\gamma,t}^{*} $ and $ \bm{F}^*_{\theta,\gamma,t+1} = \mathcal{P}^f \bm{\mu}_{\theta,\gamma,t}^{*} $.

We then match the observed and model-implied conditional distributions by minimising the conditional Kullback-Leibler divergence (CKL) over all $ T $ periods:
\begin{equation}
    \label{eq:optim1}
    \widehat{\theta}_{\gamma} =\arg\min\limits_{{\theta} \in \Theta} CKL({\theta}, {\gamma}) + \alpha \ell(\gamma)%
\end{equation}
where $ \ell(\gamma) $ is a regularisation with respect to the loss of the inner objective, penalising deviations from the model's equilibrium constraints. 
\[
\mathrm{CKL}(\theta,\gamma)
= \sum_t \left[ \sum_m
\widehat{\bm{M}}_{m,t}\,
\mathrm{KL}\!\big(\widehat{\bm{\mu}}_{\bm\cdot|m,t}\,\|\,\bm{\mu}_{\bm\cdot|m,t}(\theta,\gamma)\big)
+ \sum_f
\widehat{\bm{F}}_{f,t}\,
\mathrm{KL}\!\big(\widehat{\bm\mu}_{m|\bm\cdot,t}\,\|\,\bm\mu_{m|\bm\cdot,t}(\theta,\gamma)\big) \right],
\]
where 
\[
\bm\mu_{f|m,t}
= \frac{\bm\mu_{mft}}{\bm{M}_{mt}},
\qquad \bm\mu_{m|f,t}
= \frac{\bm\mu_{mft}}{\bm{F}_{ft}},
\]
and $ KL(a \| b) \equiv \sum_k a_k \log\frac{a_k}{b_k}$  is the standard Kullback-Leibler divergence under measure $ a $, in our case the observed conditional distribution of couplings.
We minimise the distance of the model implied conditional distribution to the observed conditional distribution, weighted by the empirical marginals.

Minimising the KL divergence of conditional distributions has the advantage that such an objective is robust to estimation errors of the marginals which are governed by the Markov transitions and not part of the policy function.
While the transition operators are unbiased and consistent estimators, this does not guarantee that within any two periods the marginals are correctly carried forwards.
Conditioning avoids that any lacking or excess mass biases the estimation our policy function.

What happens if we don't regularise? We have three tuning parameters $ w_1 $, the weight of the Bellman residual, $ w_2 $, the weight of the first order condition residual, and $ \lambda $, the weight of the regularisation in the outer loop.
If we set $ \alpha = 0 $ and $ w_2 = 0 $, we only minimise the Bellman residual in the inner loop and fit the model to the data in the outer loop.
In this case the (approximated) policy function only enters the inner problem through it's next period state, but its gradient is not part of the optimisation.
This is equivalent to inverse reinforcement-learning \citep{Ng2000} with an approximate the value function $Q_{\theta, \gamma} $ and reward $ r(\theta) $ for which the outer objective becomes $ \min_{\theta} KL(\theta, \gamma^*(\theta)) \quad \text{subject to} \quad R_{\theta, \gamma^*(\theta)} = 0 $.
This allows the outer minimiser to pick $ \theta $ that pushes the inner model outside the matching models' structure while still maintaining dynamic consistency through the Bellman equation.
We can thus interpret $ \lambda $ as a regularisation parameter that allows us to fine-tune the stringency of the model structure, and thus control a bias-variance trade-off.

\subsection{Connecting the Structural Model to Reduced Form Hazards}
\label{app:hazard}

\begin{lem}[Hazards and Surplus]
\label{lem:hazard}
Fix $t$ and a married marital-role type $mf\in\{\mathsf{CH},\mathsf{CW}\}^2$.
Let $\lambda_{0,mf},\lambda_{1,mf}\in(0,1)$ denote the per-period divorce
hazards for type $mf$ under the high-cost regime ($d=0$) and the low-cost
regime ($d=1$), respectively, and let
$\bar\sigma_{mf}\equiv\sqrt{\sigma_m^2+\sigma_f^2}$ denote the composite
dispersion.
Then the decline $\psi_{mf}$ in the net surplus of remaining married under
the low-cost regime, as in $\Phi_{mf}(d)=\Phi_{mf,0}-d\,\psi_{mf}$, satisfies
\[
\psi_{mf}\;\approx\;
\bar\sigma_{mf}\left[
\log\frac{\lambda_{1,mf}}{1-\lambda_{1,mf}}
-\log\frac{\lambda_{0,mf}}{1-\lambda_{0,mf}}
\right].
\]
\end{lem}

\begin{proof}
Let $\Delta\Phi_{mf}(t,d)=\Delta\Phi_{mf,0}(t)-\psi_{mf}\,d$ denote the net
surplus advantage of remaining married in constellation $mf$ under regime $d$.
Under transferable utility, the couple divorces
iff the joint net gain from staying is negative. With logistic shock
differences on each side scaled by $\sigma_m,\sigma_f$, the divorce event is
\[
\Delta\Phi_{mf}(t,d)+\sigma_m \xi_{it}+\sigma_f \xi_{jt}<0.
\]
Approximating $\sigma_m \xi_{it}+\sigma_f \xi_{jt}$ by a single logistic term
with matched variance, $\bar\sigma_{mf}\xi_t$ where
$\bar\sigma_{mf}=\sqrt{\sigma_m^2+\sigma_f^2}$, yields
\[
\lambda_{mf}(t,d)\approx
\Lambda\!\left(-\frac{\Delta\Phi_{mf,0}(t)-\psi_{mf}\,d}{\bar\sigma_{mf}}\right).
\]
Taking log-odds and differencing between $d=1$ and $d=0$ gives
\[
\log\frac{\lambda_1}{1-\lambda_1}-\log\frac{\lambda_0}{1-\lambda_0}
\approx \frac{\psi_{mf}}{\bar\sigma_{mf}},
\]
which rearranges to the expression in the lemma.
\end{proof}

\noindent \textbf{Note:} For small hazards, the log-odds ratio approximates the log-hazard ratio, in which case this simplifies to: 
\begin{equation}
  \psi_{mf} \approx \bar{\sigma}_{mf} \left[ \log\frac{\lambda_{1,mf}}{1-\lambda_{1,mf}} - \log\frac{\lambda_{0,mf}}{1-\lambda_{0,mf}} \right] \approx  \bar\sigma_{mf} \log\left(\frac{\lambda_{1,mf}}{\lambda_{0,mf}}\right).
\end{equation}

\noindent \textbf{Note:} Hazard ratios are identified from population expectations. Because the costs enter surpluses as parameters, in order to make individual-level statements, we rely on the monotonicity assumption on the stay-premium $ \Psi $ stating that nobody prefers paying a higher cost over a lower cost.

The RDD evaluates divorce incidence in marriages observed within a three-year window around the ban, a restriction that keeps the sample relatively homogeneous and lets us abstract from compositional change. The unit of observation is divorce for marriage $i$ in period $t$, where a period is a three-month interval. These intervals are shifted one month back from standard calendar quarters, so that the reform date (1 March 2009) is the first day of the first post-reform interval.
We exclude the first interval following the ban (donut), due to distortions of the divorce incidence caused by post-reform processing delays.
Letting $d_t \equiv \mathbf{1}\{t < c\}$ be the availability of flash divorce before the ban date $c$ we specify a logit model as a linear spline around $ c$:
\begin{equation}
\Pr(y_{it} = 1 \mid \bm{d}_{t}, \bm{q}_{t}) = \beta_0 + \beta_1\,d_t + \beta_2\,(t - c)
  + \beta_3\,d_t\,(t - c) + \bm{q}_t^\top\bm{\beta}_4,
\end{equation}
where $\bm{q}_t$ is a vector of quarterly dummies accounting for seasonality.
The policy effect is $\beta_1$, the discrete change in divorce propensity at the threshold.
We repeat this estimation for each specialisation type. 

\section{Appendix: Supplementary Figures and Tables}

\setcounter{table}{0}
\renewcommand{\thetable}{C\arabic{table}}

\setcounter{figure}{0}
\renewcommand{\thefigure}{C\arabic{figure}}

\begin{figure}[h!]
  \centering
\caption{Divorce incidence in the Netherlands, conditional on age, by gender}\label{fig_age_dif} 
\includegraphics[trim={0 0 0 0},clip,width=0.85\textwidth]{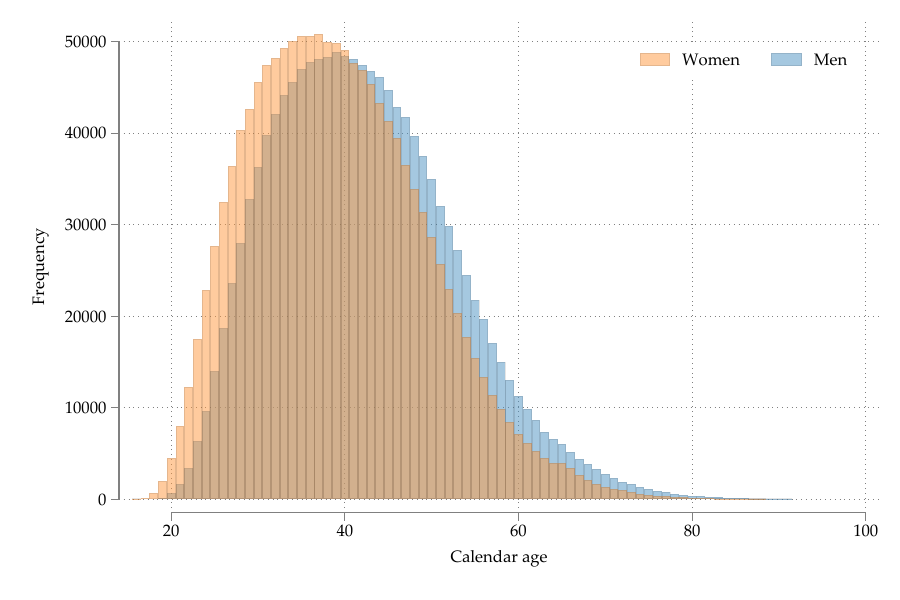} 
\begin{tablenotes}
Data from Statistics Netherlands (CBS). The histogram plots the age distribution of men and women who divorced between years 1995 and 2020 .
\end{tablenotes}
\end{figure}

\begin{table}[h!]
\centering
\caption{Association between flash divorce and demographics}
\begin{tabular}{lccc}
\hline
\textbf{Flash Divorce} & \textbf{Coefficient} & \textbf{Std.Err.} \\
\hline
Native + 1st gen immigrant & -0.0615***                 & (0.0019) \\
Native + 2nd gen immigrant & -0.0103***                 & (0.0020) \\
Both 1st gen immigrants & -0.1060***                    & (0.0019) \\
1st + 2nd gen immigrant & -0.0923***                    & (0.0035) \\
Both 2nd gen immigrants & -0.0520***                    & (0.0060) \\
Age husband             & \textcolor{white}{-}0.0006*** & (0.0002) \\
Age wife                & -0.0007***                    & (0.0001) \\
Age differential        & -0.0004*\textcolor{white}{**} & (0.0002) \\
years married           & -0.0014***                    & (0.0001) \\
one post-sec degree     & \textcolor{white}{-}0.0179*** & (0.0035) \\
both post-sec degrees   & \textcolor{white}{-}0.0306*** & (0.0033) \\
education missing (at least partially) & \textcolor{white}{-}0.0284*** & (0.0029) \\
husband breadwinner     & \textcolor{white}{-}0.0146*** & (0.0025) \\
wife breadwinner        & \textcolor{white}{-}0.0114*** & (0.0029) \\
Both employed           & \textcolor{white}{-}0.0063*\textcolor{white}{**} & (0.0025) \\
Income husband          & \textcolor{white}{-}0.7320*** & (0.0540) \\
Income wife             & \textcolor{white}{-}0.0021*** & (0.0001) \\
Income differential     & -0.0012***                    & (0.0001) \\
Constant                & \textcolor{white}{-}0.0991*** & (0.0050) \\
\hline
Observations            & 287,746 \\
\hline
\end{tabular}
\label{tab:flashreg}
\begin{tablenotes}
Standard errors in parentheses, *** $p<0.01$, ** $p<0.05$, * $p<0.1$
\end{tablenotes}
\end{table}

\begin{table}[H]
\centering
\caption{Surplus Matrix $\Phi(t, d)$: Specification with time trends}
\begin{tabular}{@{}c lccc@{}} 
\label{Table_3x3t}
 &  & \multicolumn{3}{c}{Women} \\
 \cline{3-5}
 &  & unmarried & married, home prod. & married, work \\
\cline{2-5}
\multirow{3}{*}{\rotatebox[origin=c]{90}{Men}} 
 & unmarried & $-0.91 + 0.54 t$ & $\cdot$ & $\cdot$ \\[1ex]
 & married, home prod. & $\cdot$ & $ 1.53 -0.06 t - 0.12 d$ & $ 0.29 - 0.03 t - 0.00 d$ \\[1ex]
 & married, work & $\cdot$ & $ 1.59 + 0.11 t - 0.00 d$ & $ 1.95 + 0.04t - 0.06 d$ \\[1ex]
\cline{2-5}
\end{tabular}
\begin{tablenotes}
  Scales: $ \log(\sigma) = (0, -1.55, -1.62) $, loss: $ D_{\text{KL}} = 0.21 $, time trend: $ t \in [0,1] $. 
\end{tablenotes}
\end{table}

\end{document}